\documentclass[10pt]{article}

\usepackage[preprint]{tmlr}

\usepackage[utf8]{inputenc}
\usepackage[T1]{fontenc}
\usepackage{amsmath,amssymb,amsthm}
\usepackage{booktabs}
\usepackage{xcolor}
\usepackage{graphicx}
\usepackage{tikz}
\usetikzlibrary{arrows.meta,positioning,calc,fit,backgrounds}
\usepackage{pgfplots}
\pgfplotsset{compat=1.16}
\usepackage{hyperref}
\hypersetup{colorlinks=true,linkcolor=blue!50!black,citecolor=blue!40!black,urlcolor=blue!50!black}

\newcommand{\Leray}{\mathbb{P}}
\newcommand{\Om}{\Omega}
\newtheorem{proposition}{Proposition}

\definecolor{teal}{HTML}{0B6E6E}

\title{Solver Exactness, Learned Flexibility:\\
Equivariant Boundary-Correction Operators for Stokes Flow}

\author{\name Denis Gueyffier \\ \addr Direction Scientifique G\'en\'erale, ONERA, Institut Polytechnique de Paris, Palaiseau, France}

\begin{document}
\maketitle

\begin{abstract}
Computing the viscous Stokes flow around a shape requires solving a boundary-integral equation, and for a new shape the solve must begin from scratch. Learned operators promise to spread this cost across shapes, but it is not clear what such an operator retains of the solver it replaces, or what determines whether it transfers to shapes it was not trained on. We make both questions answerable by choosing a problem which is exactly solvable except for a single term: a second-kind boundary-integral problem is solved exactly using a kernel-independent fast summation, while the boundary correction, which has no closed form, is learned. The solver's guarantees carry over unchanged: exactness on the closed-form part, $O(N)$ scaling, $SO(3)$-equivariance to machine precision, and an $O(N)$ differentiable adjoint.

We then make precise what the learning contributes. It does not contribute to accuracy, differentiability, or scaling $O(N)$, all of which are provided by the solver. Learning contributes  a one-time cost in that the forward map is trained once and then evaluated on a new shape in a single pass rather than resolved. Measured against baselines, the learned map is $5$ to $16\times$ more data-efficient than a black-box DeepONet and maintains a $\sim\!2.5\times$ lower in-distribution error than a geometry-aware operator, although that operator is stronger in the low-data limit and the learned map is less reliable out of distribution. We trace this fragility to the global parameterization and reduce it with a local equivariant kernel. The exactly-solvable setting yields clarity about the mechanism: geometric generalization is governed by invariance and coverage, not by conditioning or by capacity.
\end{abstract}

\section{Introduction}
In the Stokes regime, shape design, many-particle suspensions, and microswimmer locomotion all ask for the same flow to be computed across many shapes or many bodies. A boundary-integral solver answers each query exactly and certifiably but pays this cost again for every new geometry; a learned surrogate answers cheaply, but degrades unpredictably once the geometry leaves its training regime. We take neither side of this trade-off. Where the governing operator is known, we use it exactly, with the solver's guarantees; where it is not, we learn it. For most PDEs this is hard to act on, because it is unclear which part of the solution operator counts as ``known''. Incompressible Stokes flow is an exception, and the repeated evaluation regime above is what makes the split worth taking: an operator trained once and then applied to every shape replaces the per-geometry solve with a single one-time cost. This blend of exact structure and repeated evaluation constitutes our testbed.

The incompressible Stokes equations carry a sharp structural fact that operator learning rarely
exploits. The operator that enforces the divergence-free constraint, the Leray projector $\Leray$, is, in free space, governed by a \emph{single} kernel: the Stokeslet, or Oseen tensor, fixed by the
equations and equivariant under rotation and translation. \emph{Nothing about it depends on the data.} On
this kernel the boundary-integral method gives a certifiable solver: exact to machine precision, $O(N)$
with fast multipole acceleration \citep{greengard1987}, and reproducing the analytic Green's function to the quadrature limit.
The data dependence enters in exactly one place: the domain boundary. In a bounded geometry the
projector is the free-space operator plus a boundary correction with no closed form.

This is the opposite of how neural operators are usually built. A neural operator maps an input function
to a solution function as a composition of integral layers, each with a \emph{learnable} kernel
\citep{li2020gkn,kovachki2023no}, translation-invariant in the Fourier neural operator
\citep{li2020fno}, locally supported in graph neural operators, hierarchical in the multipole graph
neural operator \citep{li2020mgno}, so the whole solution operator, core included, is fit from data.
Here the core is not fit at all: it is known, and known exactly. This leaves a precise division
of labor, which is the inductive bias we study (Figure~\ref{fig:pipeline}):
\begin{quote}
\emph{Build the free-space core exactly and equivariantly; learn only the boundary correction.}
\end{quote}

Our setup goes beyond a clean testbed: boundary-integral methods are the standard solver for Stokes flows (microfluidics, particulate suspensions, and microswimmer hydrodynamics) where the linearity and the known Stokeslet make them the natural tool. Our paper makes the split operational and then stress-tests it. Our contributions:
\begin{enumerate}\itemsep1pt
\item \textbf{A rigid-core / learned-boundary inductive bias} for elliptic PDE operators
(Section~\ref{sec:split}), with the free-space far-field built as steerable fast multipole operators (rotation-equivariant by a fixed spherical-harmonic basis; \citealp{weiler20183d})
which are \emph{equivariant by construction}, exact to machine precision. This is in contrast to enforcing equivariance by data augmentation.
\item \textbf{Conditioning is not the bottleneck} (Section~\ref{sec:negative}): reformulating the
learned boundary correction from an ill-conditioned first-kind to a well-conditioned second-kind
Fredholm operator does \emph{not} improve learnability. Conditioning is not
the bottleneck and to our knowledge this reformulation has not been isolated before.
\item \textbf{Geometric generalization, validated and bounded} (Section~\ref{sec:gen}): a
geometry-conditioned learned operator generalizes to unseen \emph{regular} shapes within the training
envelope, matching the exact reference in the interior; it fails to extrapolate beyond the envelope.
\item \textbf{Invariance is the bottleneck} (Section~\ref{sec:invariance}): we isolate, by ablation,
that cross-shape generalization is governed by the \emph{equivariance of the shape descriptor}, not
by capacity and not by the method. A noninvariant descriptor collapses under rotation; a network does
not rescue it. This connects geometric OOD in operator learning to the canonicalization approach
\citep{shumaylov2025lie,tahmasebi2025bounds,adaptivecanon2025}.
\item \textbf{Opening the exterior problem in 3D} (Section~\ref{sec:exterior}): we carry the split to
the exterior setting: flow past a body, drag and mobility. The operator foundation is a 3D completed
double layer made exact on the surface by quadrature by expansion, second-kind well-conditioned
\emph{and} $SO(3)$-equivariant by construction. Its drag matches the analytic Oberbeck solution. On top
of it we learn drag, many-body mobility, and the exterior field. Equivariance is exact by construction
(against $4\times10^{-2}$ for rotation augmentation) and transfers without canonicalization; Lorentz
reciprocity of the two-body resistance is recovered to machine precision; the out-of-distribution
fragility is confined to scalar invariants rather than to structure.
\end{enumerate}
The core results are on a controlled 2D interior Stokes testbed where an exact boundary-integral solve
provides ground truth (the dense exact solve \emph{is} the BEM reference), in double precision, on CPU;
Section~\ref{sec:exterior} then opens the exterior problem in three dimensions, with the analytic sphere
and ellipsoid drag as reference.

\begin{figure}[t]
\centering
\resizebox{\textwidth}{!}{%
\begin{tikzpicture}[font=\sffamily\scriptsize,
 exact/.style={draw=blue!45!black, line width=0.7pt, rounded corners=2pt, fill=white, align=center, inner sep=3pt},
 learn/.style={draw=orange!80!black, line width=0.7pt, rounded corners=2pt, fill=white, align=center, inner sep=3pt},
 io/.style={draw=black!55, line width=0.7pt, rounded corners=2pt, fill=white, align=center, inner sep=3pt},
 proc/.style={draw=black!60, line width=0.7pt, rounded corners=2pt, fill=black!3, align=center, inner sep=5pt, font=\sffamily\footnotesize},
 ar/.style={-{Stealth[length=5pt,width=4pt]}, line width=0.8pt, black!70}]
\node[io] (dom) at (0,0)
 {\includegraphics[width=18mm]{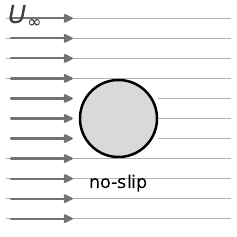}\\[1pt]body in stream $U_\infty$\\ no-slip $\partial B$};
\node[exact] (core) at (3.7,1.25)
 {\includegraphics[width=18mm]{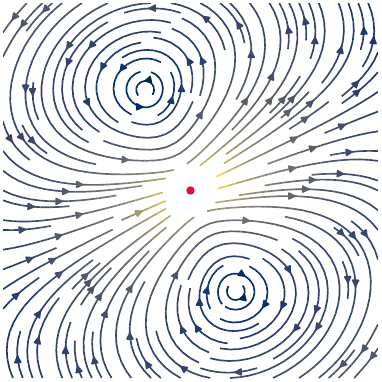}\\[1pt]free-space core\\ (Stokeslet) $\cdot$ \textbf{0} params};
\node[learn] (corr) at (3.7,-1.25)
 {boundary correction $K_\theta$\\ \textbf{learned}};
\node[font=\LARGE] (plus) at (6.35,0) {$\boldsymbol{\oplus}$};
\node[proc] (solve) at (8.0,0) {completed\\ second-kind\\ solve};
\node[io] (out) at (10.6,0)
 {\includegraphics[width=18mm]{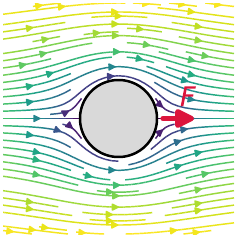}\\[1pt]exterior field $u\!\to\!U_\infty$\\ drag, mobility $\mathcal{J}$};
\draw[ar] (dom.east) -- (core.west);
\draw[ar] (dom.east) -- (corr.west);
\draw[ar] (core.east) -- (plus);
\draw[ar] (corr.east) -- (plus);
\draw[ar] (plus) -- (solve.west);
\draw[ar] (solve.east) -- (out.west);
\draw[ar, dashed, black!55] (out.south) .. controls (10.6,-2.7) and (3.7,-2.7) .. (corr.south)
 node[pos=0.5, below=1pt, black!60, font=\sffamily\scriptsize] {differentiable: gradients by adjoint, $O(N)$};
\begin{scope}[on background layer]
 \node[fill=blue!7, rounded corners=3pt, fit=(core), inner sep=4pt] (gco) {};
 \node[fill=orange!9, rounded corners=3pt, fit=(corr), inner sep=4pt] (gcr) {};
\end{scope}
\node[blue!45!black, font=\sffamily\scriptsize, above=2pt of gco] {fixed analytically};
\node[orange!75!black, font=\sffamily\scriptsize, below=2pt of gcr] {the only learned part};
\end{tikzpicture}}
\caption{The rigid-core / learned-boundary split, shown on the \emph{exterior} problem.
A rigid body $B$ held in a uniform stream $U_\infty$ with no-slip on $\partial B$ (\textbf{left}; a sphere
here, the construction is shape-general) drives a disturbance velocity that decomposes into an exact
\emph{free-space core}, the Stokeslet, $SO(3)$-equivariant and with no learnable parameters
(\textbf{top, blue}), plus a \emph{boundary correction} $K_\theta$ enforcing no-slip on $\partial B$
(\textbf{bottom, orange}); \emph{only this correction is learned}, as a well-conditioned second-kind
operator equivariant by construction. A completed second-kind solve combines them ($\boldsymbol{\oplus}$)
into the exterior field (\textbf{right}, tending to $U_\infty$ at infinity) and the quantities of interest, drag, torque, and mobility $\mathcal{J}$. The same split governs the controlled interior testbed where
an exact dense solve furnishes the ground truth for the experiments below. The whole pipeline is
differentiable, so shape gradients follow from a single adjoint solve at $O(N)$.}
\label{fig:pipeline}
\end{figure}
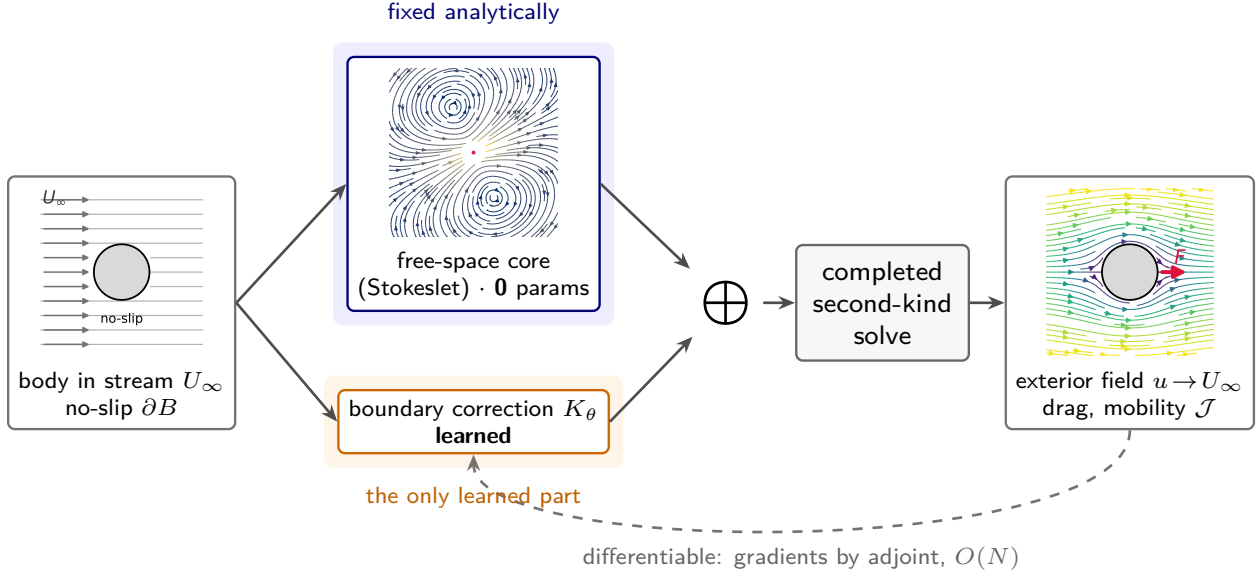

\section{Experimental summary and metrics}\label{sec:expsummary}
No single number captures where this method stands, so Table~\ref{tab:positioning} reports a profile rather than a ranking. The learned operator inherits the solver's guarantees (exactness, $O(N)$ scaling, exact $SO(3)$-equivariance, a differentiable adjoint) and trains on an order of magnitude fewer shapes than a black-box surrogate, while conceding the two axes it does not lead, out-of-distribution robustness and raw per-query speed. It is not Pareto-dominated by any single alternative. Only our column and the reported ratios are measured here; the competing entries are qualitative, set by what each method is rather than by a head-to-head run.

\begin{table}[t]
\centering
\caption{Positioning across metrics, \emph{not} a speed benchmark. Only the ML-FMM column and the ratios ($5$--$16\times$, $\sim\!1.5\times10^{-3}$, $\sim\!2.5\times$ in-distribution, $10^{-16}$ vs $4\times10^{-2}$) are measured on our testbed; competing columns are qualitative, set by the nature of each method. The method is not Pareto-dominated: it leads on guarantees, data efficiency, and exact equivariance, and concedes out-of-distribution robustness and raw per-query speed.}
\label{tab:positioning}
\resizebox{\textwidth}{!}{%
\begin{tabular}{lcccc}
\toprule
Axis & Exact BIE solver & DeepONet & Geometric NO & ML-FMM (ours) \\
\midrule
Tensor error & exact (ref.) & high & moderate & $\sim\!1.5\times10^{-3}$ \\
Data efficiency & n/a & $5$--$16\times$ vs ours & $>$ ours & $\mathcal{O}(10)$ shapes \\
$SO(3)$-equivariance & exact & none & approx.\ ($4\times10^{-2}$) & exact ($10^{-16}$) \\
Differentiable adjoint $O(N)$ & yes, re-solved each time & yes, no guar.\ & yes, no guar.\ & yes, exact; trained once for all shapes \\
Scaling & $O(N^2)$ or $O(N)$ & $O(N)$ & $O(N)$ on grid & $O(N)$ (FMM) \\
OOD robustness & exact everywhere & weak & variable & fragile (local kernel helps) \\
Per-query cost & slow ($\sim\!0.45$\,s) & fast & fast & fast ($\sim\!50\,\mu$s) \\
\bottomrule
\end{tabular}%
}
\end{table}

Three error metrics recur; we define them once here and use them throughout. \emph{Held-out no-slip} is
the relative $L_2$ residual of the no-slip condition on shapes not seen in training,
$\|u_{\mathrm{core}}+u_{\mathrm{corr}}\|_{L_2(\partial\Omega)}/\|u_{\mathrm{core}}\|_{L_2(\partial\Omega)}$
(dimensionless; $\sim\!10^{-15}$ for the exact solve, $\sim\!10^{-3}$ for a learned correction at the
envelope floor). \emph{Interior field $L_2$} is the relative error of the predicted interior velocity at
fixed query points on held-out shapes, $\|u_{\mathrm{pred}}-u_{\mathrm{ref}}\|_{L_2}/\|u_{\mathrm{ref}}\|_{L_2}$
(the head-to-head metric; floor $\sim\!10^{-3}$). \emph{Out-of-distribution error distribution} reports the
median, p90, and maximum of the interior-field error over shapes drawn beyond the training envelope, because
out-of-distribution failure is heavy-tailed and the tail, not the mean, is decisive. Table~\ref{tab:expsummary}
consolidates every experiment with its setting, metric, and headline result.

\begin{table}[t]\centering\small
\caption{Consolidated summary of all experiments: setting, metric, and headline result. Section references point
to the full protocol, figures, and tables.}\label{tab:expsummary}
\begin{tabular}{p{2.5cm}p{2.9cm}p{2.0cm}p{4.0cm}}\toprule
Experiment (\S) & Setting & Metric & Headline result \\\midrule
Conditioning (\S\ref{sec:negative}) & interior, one curve, 1st vs 2nd-kind learned & held-out no-slip & 2nd-kind no better for \emph{learning}: representability, not conditioning, is the floor \\
Generalization (\S\ref{sec:gen}) & bounded interior, broad shapes & held-out no-slip & interpolation holds; extrapolation does not \\
Invariance (\S\ref{sec:invariance}) & cross-family & held-out no-slip & invariance, not capacity, is the lever; $>\!10^5\times$ collapse under rotation without it \\
Coverage (Tab.~\ref{tab:coverage}) & continuous broad dist.\ & held-out no-slip & coverage breaks the cross-family plateau ($\approx\!3.6\times10^{-2}$) \\
Ablation (Tab.~\ref{tab:ablation}) & smooth family & held-out error & each ingredient necessary \\
Data efficiency (Tab.~\ref{tab:headtohead}) & shape$\to$field vs DeepONet & interior field $L_2$ & $5$--$16\times$ more data-efficient than the black box \\
Geometry-aware (\S\ref{sec:headtohead}) & vs GNO/GINO kernel & interior field $L_2$ & $\sim\!2.5\times$ in-distribution (crossover at low data); split more fragile OOD \\
OOD parameterization & amp.\ $0.40$, end-to-end & field error dist.\ & local equivariant kernel recovers the exact operator; dense global map develops a heavy tail \\
QBX (Fig.~\ref{fig:qbx}) & on-surface 3D & rel.\ error vs order & geometric convergence in expansion order \\
Exterior 3D (\S\ref{sec:exterior}) & sphere / ellipsoid drag & drag vs $6\pi\mu a$ / Oberbeck & $\sim\!10^{-3}$; $SO(3)$-equivariant to $10^{-14}$ \\
Learned resistance (\S\ref{sec:exterior}) & $R(S)$ on $\ell{=}2$ family & equivariance; held-out & $2\times10^{-16}$ equivariance; $\sim\!1.5\times10^{-3}$ held-out, $\sim\!3\times10^{-3}$ OOD \\
Two-body (\S\ref{sec:exterior}) & two spheres & reciprocity; drag & resistance symmetric to $6\times10^{-16}$; reflections $40\%\!\to\!1.3\%$ \\\bottomrule
\end{tabular}
\end{table}

\section{Background and related work}
\paragraph{Operator learning as kernel integration.}
\citet{li2020gkn,kovachki2023no} cast operator learning as stacked kernel integrals; \citet{li2020fno}
(FNO), DeepONet \citep{lu2021deeponet}, and the multipole graph neural operator \citep{li2020mgno}
differ in the kernel's structure. Approximation guarantees are given by \citet{kovachki2023no}. Our
view is complementary: rather than learning the whole kernel, we fix the part the PDE already
determines and learn only the remainder.

\paragraph{Learning Green's functions and learning on the boundary.}
A growing line learns Green's functions directly \citep{boulle2022green,neuralgreens2025} or formulates
operator learning through boundary integrals. Closest to us, \citet{fang2023boundary} learn
\emph{only on the boundary} for parametric PDEs in complex geometries; \citet{fieno2024} build a
Fredholm-integral neural operator for boundary-value problems. We share the
boundary-only and second-kind ingredients, but our emphasis is the \emph{rigidity of the free-space
core}, with nothing to learn, and a finding on what the second-kind buys for \emph{learning}.
Closest in spirit, \citet{han2026kernel} read geometric generalization through the same kernel-integral
lens and draw the same bridge to fast summation, but \emph{learn} the geometry-dependent kernel with a
multiscale, Ewald-inspired operator carrying accuracy guarantees. The contrast sharpens our thesis: they learn the kernel integral; we show the free-space kernel \emph{is} the exact
Stokeslet (rigid, zero parameters) and confine learning to the boundary remainder.

\paragraph{Equivariant operators and the role of boundaries.}
Equivariance is a standard inductive bias \citep{weiler20183d}; spherical/steerable operators enforce
it for efficiency \citep{bonev2023sfno,cliffordcond2025}. Recent work observes that strict equivariance is
\emph{insufficient} for physical systems precisely because boundaries break it, and proposes additive
splits that decouple a globally equivariant component from boundary perturbations
\citep{gsno2026}. Our construction is exactly such a split made exact: an equivariant free-space core
plus a nonequivariant, learned boundary correction. Empirically we contrast such priors against augmentation directly (Section~\ref{sec:exterior}): the exact
tensorial prior reaches machine-precision equivariance ($2\times10^{-16}$) where rotation augmentation of the
\emph{same} network plateaus at $4\times10^{-2}$, a $10^{14}$ gap, so the structural constraint, not added
capacity or data, is what buys exact invariance.

\paragraph{Geometric generalization and canonicalization.}
Geometry-aware operators encode shape via signed distance functions, point clouds, or attention over learned physical states
\citep{zhong2025pigano,ginot2025,wen2025gaot,wu2024transolver,flowbench2025}, generalizing across geometries with a single forward
pass. These methods learn the \emph{entire} solution operator through a learned geometry encoder; by contrast, our
split fixes the free-space core analytically (zero learnable parameters, exact equivariance) and confines
learning to the boundary correction. We study this choice in a controlled 2D setting rather than at the
industrial 3D scale where such surrogates now excel \citep{wen2025gaot}. A companion study (the Hybride
axis of this programme) reaches a convergent finding from the CFD side: a non-local/hybrid geometry
encoder generalizes better out of distribution than a purely local (co-located SDF) one, matching the
structural lesson here that what is \emph{learned} (the boundary correction), not what is \emph{fixed}
(the exact core), governs out-of-distribution behaviour. Whether such generalization \emph{extrapolates} is delicate: \citet{dulberg2020canonical} find
training-set diversity suffices for translation but rotation is far harder; canonicalization renders
operators equivariant under group actions \citep{shumaylov2025lie}, with generalization bounds
\citep{tahmasebi2025bounds}. Our invariance finding (Section~\ref{sec:invariance}) is an instance of
this principle inside boundary-integral operator learning.

\paragraph{Locality as the route to out-of-distribution geometric transfer.}
A concurrent and fast-growing line, largely post-dating our experiments, isolates \emph{locality}
as the mechanism for out-of-distribution geometric generalization. Meshfree exterior-calculus operators
built on an invariant local frame transfer from a single solution to unseen geometries, with a solution
error bound that splits into discretization and kernel approximation terms \emph{independent of the
problem geometry} \citep{meec2026}; local discrete-stencil operators generalize across shapes by
construction, mirroring classical finite-difference/element discretizations \citep{disol2026}; and, for
Stokes specifically, learned micro-solvers that depend only on the \emph{local wall geometry} come with a
guarantee that a bounded training loss yields a bounded macroscopic error \citep{hmm2025stokes}. We read
this convergence as independent corroboration of the mechanism we isolate in
Section~\ref{sec:invariance}, not as a competing claim. Our distinction is not locality \emph{per se}: we
do not learn the local operator at all. The free-space core is the \emph{exact} analytical Stokeslet/FMM
kernel (zero parameters, machine precision, $O(N)$) and only the boundary closure is learned (the unknown remainder closing the otherwise exact split, in the sense of learned closure models; \citealp{duraisamy2019}). Where
these methods replace the entire local solve with a network (and inherit its approximation floor), we
retain a certifiable numerical-analysis backbone, a bounded-domain Leray projector that matches the
analytic Green's function to the quadrature limit, and confine learning to the irreducible boundary
remainder.

\section{The rigid core and the learned boundary}\label{sec:split}

The split has two parts, one fixed and one learned, sketched in Figure~\ref{fig:pipeline}.
\paragraph{Free space: a single equivariant kernel.}
For Stokes flow the velocity due to a force is the Stokeslet $G$, and the Leray projector applied to a
field is a convolution against a kernel fixed by the PDE. In complex variables the 2D Stokeslet has a
compact $(z,\bar z)$ form whose multipole and local expansions are steerable: the three fast multipole
translations (M2M, M2L, L2L) are, for each multipole frequency, scalar multiples of a fixed
rotation-equivariant map. We express the far-field operators in this fixed steerable basis; the scalars are
\emph{determined analytically by the kernel} (e.g.\ binomial weights for the translations), not free
parameters, so the core carries \emph{no geometry- or data-dependent degrees of freedom}. The operators
are exactly representable in the basis and equivariant to machine precision by construction (residual
$\sim\!10^{-16}$). To make ``nothing to learn'' precise rather than rhetorical, we tried to \emph{fit}
the operators from field-reconstruction data: the loss is underdetermined and plateaus at
$\sim\!8\times10^{-3}$ without recovering the analytic operators (an interoperator gauge degeneracy).
Constraining the fit to match the operators directly (distillation) recovers them to
$4.7\times10^{-6}$ (two-stage) and $9.9\times10^{-6}$ (L-BFGS, no field data). The optimum simply
\emph{is} the analytic operator: the free-space core is \emph{rigid}, with zero learnable content. Any ``fitting'' there is at best a distillation of a closed form, never a source of new degrees of freedom.

\paragraph{Bounded domain: the boundary correction is the only learnable object.}
In a bounded geometry $\Om$, the Leray projector splits into the fixed free-space convolution and a boundary correction,
\begin{equation}\label{eq:split}
\Leray_\Om \;=\; \underbrace{\Leray_{\mathrm{free}}}_{\text{exact, $0$ parameters}}\;+\;\underbrace{\mathcal{C}_{\partial\Om}}_{\text{learned}}.
\end{equation}
The correction $\mathcal{C}_{\partial\Om}$ enforces the no-slip boundary condition and has no closed form; it is the natural target of learning and, we will
argue, the only appropriate one. We represent it as a boundary-integral operator on
the curve and learn its geometry-dependence.

\paragraph{Second-kind formulation.}
A first-kind (single-layer) representation is ill-conditioned (the operator is compact; its singular
values accumulate at zero). The classical remedy is a second-kind (double-layer) representation: the
boundary density $\mu$ solves
\begin{equation}\label{eq:secondkind}
\bigl(-\tfrac12 I + K_{\partial\Om}\bigr)\,\mu \;=\; g,
\end{equation}
identity plus compact, whose conditioning is bounded independently of the mesh.
On our wavy closed curve we measure, for the completed operator, $\mathrm{cond}\approx 2.85$ at all
resolutions, against $\approx\!4.3\times10^{16}$ for the first-kind. We learn only the compact part
$K_{\partial\Om}$ as a function of geometry; the $-\tfrac12 I$ in \eqref{eq:secondkind} is a fixed skip (a residual identity), so the network
never has to represent the dominant term.

\section{Conditioning is not the learning bottleneck}\label{sec:negative}
\paragraph{Hypothesis.} Since the second-kind operator is exponentially better conditioned, one
expects it to be \emph{easier to learn}, the learned operator should track the exact one more
faithfully and enforce no-slip to a lower floor. We test this directly.

\paragraph{Protocol.} On a fixed family of wavy closed curves parameterized by an amplitude $\epsilon$,
with a boundary datum \emph{realizable by construction} (the trace of an interior Stokeslet, hence in
the operator's range\footnote{A generic wall velocity lies outside the range and yields a spurious
$\sim\!5\times10^{-2}$ residual independent of resolution, a pitfall we flag explicitly.}), we learn
the $\epsilon$-parameterized boundary operator in both first and second kind, and measure held-out
no-slip on disjoint $\epsilon$.

\paragraph{Result (Table~\ref{tab:negative}).} At sufficient capacity, both kinds reach nearly the same
held-out boundary residual (a few parts in $10^{5}$); the second-kind gain is $0.9\times$, \emph{no advantage}, despite a $10^{16}\times$ better conditioning. At reduced capacity (the regime where a floor
appears), the gain is only $1.5\times$ and collapses under operator noise; combining the second-kind
with an output (no-slip) training metric yields a marginal $1.3\times$. \textbf{Conditioning is not the
bottleneck.} The floor is set by the \emph{representability} of the boundary density (how many degrees
of freedom), an object distinct from operator conditioning; improving conditioning does not attack it. This negative is \emph{deterministic} by construction (fixed shape families and closed-form least squares; identical across runs), not a seed-dependent estimate.

\begin{table}[h]\centering\small
\begin{tabular}{lccc}
\toprule
capacity & first-kind no-slip & second-kind no-slip & gain (1st/2nd)\\
\midrule
sufficient & $2.6\times10^{-5}$ & $2.9\times10^{-5}$ & $0.9\times$\\
reduced & $4.0\times10^{-3}$ & $2.6\times10^{-3}$ & $1.5\times$\\
\bottomrule
\end{tabular}
\caption{Held-out no-slip, learned first- vs second-kind boundary operator on the same curve. The
second-kind's bounded conditioning ($2.85$ vs $4.3\times10^{16}$) does not translate into better
learning.}
\label{tab:negative}
\end{table}

\paragraph{What noise and nonrealizability do.} The negative above is stated on clean, realizable data. We
now probe the two regimes one should worry about. The first-kind problem is a discrete ill-posed system: its
singular values decay to zero without a gap, so noise in the small singular directions is amplified, and the
textbook remedy is precisely the second-kind reformulation \citep{hansen1998}. Two findings result.

\emph{(i) Under operator noise on a realizable, smooth datum, the negative is robust.} The second-kind no-slip
floor simply tracks the injected noise level, uniformly across resolution, and the first-kind floor matches it
(gain $\approx\!1$). This holds even though the first-kind density norm is inflated about tenfold by the
conditioning.

\emph{(ii) Conditioning reenters once the datum stops being smoothly realizable.} As an interior source
approaches the boundary, its trace acquires fine structure that leaks into the ill-conditioned subspace. The
first-kind floor then rises above the noise level, by a resolution-dependent, high-variance factor (it depends
on how the perturbation aligns with the small singular directions), while the second-kind floor stays pinned
at the noise level. The second-kind thus regains a genuine advantage \emph{exactly} for data that probes the
small singular directions.

The finding is therefore regime-specific. Improving conditioning does not help when the
bottleneck is representability; it does help when the datum, through noise or nonrealizability, excites the
ill-conditioned subspace, as classical discrete ill-posed theory predicts. This also gives a train-time \emph{diagnostic} for which regime one is in, from quantities available before any expensive solve: compare a conditioning-limited estimate (operator conditioning $\times$ data-noise level) against the representability floor (the best-fit residual of the descriptor$\to$target regression on clean data), and take whichever dominates. The second-kind reformulation helps only when the former does.

\begin{table}[h]\centering\small
\caption{Train-time diagnostic: a conditioning-limited estimate (conditioning$\times$noise) against the clean
representability floor. Whichever dominates names the regime; the second-kind reformulation helps only when the
conditioning term does.}
\label{tab:diagnostic}
\begin{tabular}{lccc}\toprule
regime & conditioning$\times$noise & repr.\ floor & verdict \\\midrule
clean, realizable (2nd-kind) & $3.5\times10^{-12}$ & $3\times10^{-3}$ & representability floor \\
clean, realizable (1st-kind, cond$\sim$N) & $3.0\times10^{-10}$ & $3\times10^{-3}$ & representability floor \\
$1\%$ noise (2nd-kind) & $3.5\times10^{-2}$ & $3\times10^{-3}$ & conditioning-limited \\
$1\%$ noise (1st-kind, cond$\sim$N) & $3.0\times10^{0}$ & $3\times10^{-3}$ & conditioning-limited \\
\bottomrule
\end{tabular}
\end{table}

\section{Geometric generalization, validated and bounded}\label{sec:gen}
The test of the split is whether a single learned correction transfers to shapes it never saw, and how far.
\paragraph{Setup.} We move to multiharmonic closed curves $r(t)=a\,(1+\sum_k c_k\cos kt)$ parameterized
by an amplitude \emph{vector} $c$ (modes $\{2,3,4\}$), and learn the second-kind compact part as a
function of $c$ (polynomial feature map, closed-form least squares). The exact dense second-kind solve
is the reference (the BEM ground truth for this boundary-integral problem). We train on shapes with
$|c_k|\le 0.05$ and test on held-out shapes.

\paragraph{Interpolation generalizes; extrapolation does not.} On held-out shapes \emph{within} the
training envelope, the no-slip is of the same order as training (ratio $\approx\!1$, no overfitting),
and a geometry-\emph{independent} operator is $\sim\!10^{2}\times$ worse, so the learned operator has
captured the shape dependence. The \emph{interior} velocity field, not just the boundary, matches the exact reference to $\sim\!1.6\times10^{-4}$ relative $L_2$: the operator
reproduces the solution in the domain. \emph{Beyond} the envelope, with amplitudes larger than any seen in training, it degrades sharply, a relative increase of roughly $42$-fold in the boundary residual at
$2.4\times$ the training amplitude. The
honest claim is therefore \emph{generalization within the training envelope}, consistent with the
broader finding that operator-learning geometric generalization interpolates more readily than it
extrapolates \citep{dulberg2020canonical}.

\paragraph{Descriptor choice and resolution.} We separated the roles of descriptor \emph{type} (global
Fourier amplitudes vs.\ locally sampled point-radii) and \emph{dimensionality}, on both a smooth band-limited
family and a localized (von Mises bump) family, measuring held-out no-slip in interpolation and extrapolation.
Three findings follow, one counter to our initial guess.

\emph{(i) The floor is set by representability, not by the descriptor.} Once any descriptor resolves the
shape, Fourier and point-radii converge to the \emph{same} held-out floor (ratio $\approx\!1$ at matched dimension), consistent with the central thesis that capacity, not encoding, is the bottleneck.

\emph{(ii) Resolution helps interpolation, but never extrapolation.} It helps up to the resolving threshold
and then plateaus (a $\sim\!15$--$25\times$ drop from an under-resolved to a resolved descriptor). It never
closes the extrapolation gap, which stays above the interpolation floor at every dimension: resolution is not
a substitute for structure outside the training envelope.

\emph{(iii) Counter to intuition, a global descriptor is the more robust under-resolution.} One expects local
sampling to suit localized features; instead, a \emph{global} descriptor degrades gracefully while \emph{local}
point-radii can miss a localized bump entirely ($\sim\!6\times$ worse at low dimension, catastrophic in
extrapolation). Each global coefficient sees the whole contour, whereas a sparse local sample may straddle the
feature. For smooth, nonlocalized shapes the two are interchangeable even at coarse sampling.

\paragraph{General regular shapes, not just harmonic bumps.} The boundary-integral solve itself is
\emph{not} tied to a harmonic basis: discretized by collocation, it reaches machine precision on a range
of nonharmonic regular shapes, ellipses ($\mathrm{res}\sim\!4\times10^{-16}$), localized von Mises
bumps ($\sim\!5\times10^{-16}$), rounded polygons ($\sim\!4\times10^{-16}$), with bounded
conditioning ($2.6$--$9.8$). The learning descriptor likewise need not be harmonic: a point-based
descriptor (sampled boundary radii, no Fourier transform) generalizes identically to the harmonic
amplitudes. The harmonic parameterization is a convenience; it restricts to star-shaped domains, a limitation of the generator, not of the method.

\paragraph{Off the star-shaped parameterization.} To confirm the method is not tied to that star-shaped
generator, we build a smooth closed curve \emph{outside} the $r(t)$ family: a strip bent into a thick arc
(a horseshoe). We hand its nodes (positions, outward normal, tangent, curvature) directly to the same
operators, which assume nothing about star-shapedness. The shape is strongly nonconvex: $35\%$ of its
boundary is concave, qualitatively unlike the nearly convex training families.

Both structural pillars survive. The completed second-kind operator stays well-conditioned
($\mathrm{cond}\approx27$, $N$-independent), against a first-kind operator at $\sim\!10^{17}$. A manufactured
interior Stokeslet solve converges \emph{spectrally} under refinement
($2\times10^{-4}\!\to\!8\times10^{-6}\!\to\!4\times10^{-10}\!\to\!1\times10^{-11}$ as
$N=192\!\to\!256\!\to\!448\!\to\!512$), confirming the operators are correct off the parameterization.

This horseshoe is strongly nonconvex but still star-shaped, so some interior point sees the whole boundary.
Driving the geometry to \emph{strictly} non-star (a deep C with an empty visibility kernel, its two arms brought toward a near-slit) degrades the second-kind conditioning to $\sim\!10^{17}$, and it
\emph{does not improve with refinement} (a geometric, not a resolution, effect). The edge of the
well-conditioned regime is precisely the near-slit onset characterized next.

\paragraph{Piecewise-regular shapes.} Piecewise-smooth boundaries joined $C^1/C^2$, with no corner, are
handled with \emph{graded} accuracy. A $C^2$ periodic spline solves to $\sim\!10^{-5}$: the trapezoidal
quadrature loses spectral accuracy on a nonanalytic curve, but the operator stays well-conditioned
($\approx\!2.7$). Hardening toward a corner (a superellipse $|x|^p+|y|^p=1$ with growing $p$) degrades the
residual monotonically (from $10^{-16}$ for the circle to $5\times10^{-6}$ for a near-square), while
conditioning stays moderate.

A genuine corner ($C^0$) breaks both the spectral quadrature and the compactness of the double layer. We
characterize it with a panel-based second-kind solve (Laplace double layer, composite Gauss--Legendre; no
singular quadrature is needed, since the kernel vanishes on a straight panel). Two honest findings result.

\emph{Moderate corners are already robust.} A convex polygon solves the interior Dirichlet problem to
$\sim\!10^{-9}$ with uniform panels at conditioning $O(10)$, consistent with the superellipse trend above.

\emph{The genuine breakdown is the near-slit.} As the reentrant interior angle $\beta\!\to\!2\pi$, the
conditioning of $(-\tfrac12 I+K)$ blows up (from $\sim\!13$ at $\beta=1.2\pi$ to $\sim\!4\times10^{3}$ at
$\beta=1.95\pi$). Resolving the singular density there needs refinement toward the corner, but naive graded
refinement \emph{worsens} the conditioning (tiny panels). This is exactly what RCIP
\citep{helsing2008rcip,helsing2018tutorial} compresses, kernel-independently, into a coarse well-conditioned
system for transformed densities. We implemented RCIP on Helsing's model problem (Laplace adjoint double-layer, single corner) and reproduced the published functional $q_{\mathrm{ref}}=1.130016\ldots$ to $3.8\times10^{-9}$ \emph{with the conditioning of the compressed system bounded at $\approx\!1.7\times10^{3}$ for every refinement depth} (nsub $5$ to $100$). The brute-force graded system, by contrast, reaches that accuracy only transiently before its conditioning diverges past $10^{7}$ to overflow (Figure~\ref{fig:cornerrcip}A; Table~\ref{tab:rcip}).

The breakdown is governed by the \emph{local} corner angle alone, and the corner block is
\emph{learnable}: the compressed inverse $R(\theta)$
($64\times64$) interpolates across corner angles, and injecting it into the compressed solve at
\emph{genuinely held-out} angles (Chebyshev sampling, $K{=}16$) reproduces $q$ to $2.9\times10^{-5}$, improving with sampling density (Figure~\ref{fig:cornerrcip}B; Table~\ref{tab:rcip}). The corner block is thus a $\theta$-parameterized, geometry-local object, learned once and then applied across global shapes. This is exactly the theme of Section~\ref{sec:headtohead}, an operator trained once and applied to many shapes, realized at the corner. \textbf{Summary: regular and piecewise-regular shapes
are general and moderate corners are already robust; for the near-slit / many-corner regime, RCIP restores
bounded conditioning (validated to $\sim\!10^{-9}$ against a published reference) and its corner block is
a learnable $\theta$-parameterized object ($\sim\!10^{-5}$ to $\sim\!10^{-10}$ at held-out angles, improving with sampling), a concrete $\theta$-local extension, shown on the Laplace model problem, with an analytic extension to vector Stokes flow below.}

\paragraph{Vector extension: Stokes flow to the re-entrant corner.} The same construction lifts to the two-dimensional Stokes equations, the vector case with the strongest corner singularity. Under the \emph{completed double-layer} representation $\mathbf u=\mathcal D[\mathbf q]+\mathbf n(\mathbf x)\oint\mathbf q\cdot\mathbf n\,\mathrm ds$, the rank-one completion pins the interior pressure null space that the combined field $2\mathcal D+\eta\mathcal S$ leaves under-regularized (a single near-null mode aligned with $\mathbf n$), while the double layer alone is scale invariant at the corner, so the compressed inverse is again a fixed point, a pure function of $\theta$. The scalar analysis carries over directly: RCIP restores bounded conditioning where a graded direct solver of equal accuracy diverges, and $R(\theta)$ remains the same low-dimensional $\theta$-parameterized object, recovered by interpolation of the compressed inverse. The one structural difference is at the slit, where $R(\theta)$ inherits a branch point as $\theta\to2\pi$ (corner exponents colliding at $\lambda=\tfrac12$), an intrinsic bound on global polynomial interpolation in the deep re-entrant regime rather than a numerical artefact. We confirm this numerically at a right-angle corner: RCIP holds $\kappa\sim2.9\times10^7$ at fixed subdivision against a diverging graded solver, recovers a manufactured Stokeslet to $\sim\!2.6\times10^{-9}$, and interpolates $R(\theta)$ to $\sim\!6\times10^{-11}$ on the convex branch and $\sim\!4\times10^{-9}$ in the re-entrant regime at held-out angles, matching the scalar model.

\begin{figure}[t]\centering
\includegraphics[width=0.93\textwidth]{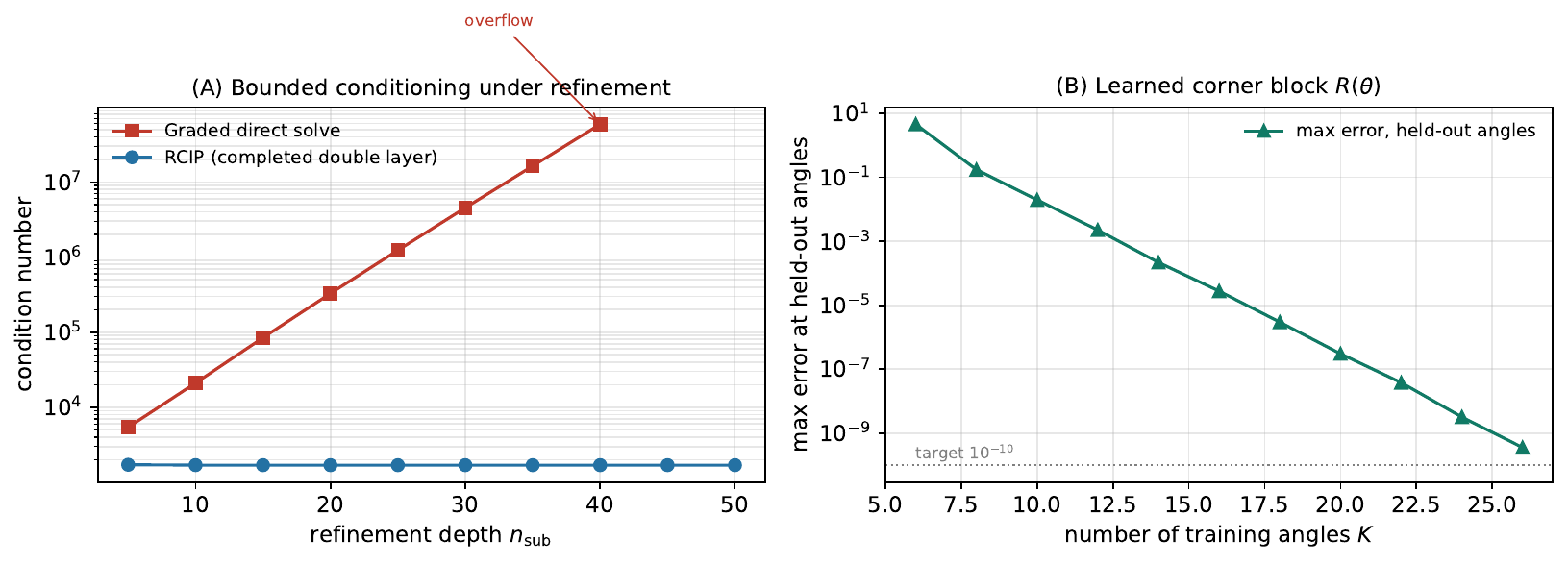}
\caption{Corner regime via RCIP on Helsing's model problem (Laplace adjoint double layer). \textbf{(A)} Under refinement toward the corner, the compressed condition number stays bounded ($\approx\!1.7\times10^{3}$) at every depth, while a graded direct solve of equal accuracy diverges past $10^{7}$ into overflow. \textbf{(B)} The learned corner block $R(\theta)$, reconstructed by barycentric interpolation of the exact compressed inverse on Chebyshev-sampled angles, converges spectrally in the number of training angles, reaching $2.8\times10^{-5}$ at $K{=}16$ and $3.6\times10^{-10}$ at $K{=}26$ on genuinely held-out angles.}
\label{fig:cornerrcip}
\end{figure}

\begin{table}[h]\centering\small
\caption{Corner regime, validated on Helsing's model problem (Laplace adjoint double layer, single corner,
published $q_{\mathrm{ref}}=1.1300163213\ldots$). RCIP attains the reference accuracy with \emph{bounded}
conditioning where the brute-force graded system diverges; the compressed corner block $R(\theta)$ is
learnable, reproducing $q$ at held-out angles to machine precision.}
\label{tab:rcip}
\begin{tabular}{lcc}\toprule
 & accuracy in $q$ & conditioning of the solve \\\midrule
brute-force graded (nsub $40$) & $8.2\times10^{-10}$ & $5.9\times10^{7}$, diverges (overflow by nsub $50$) \\
RCIP compressed (nsub $40$--$100$) & $3.8\times10^{-9}$ & $\approx 1.7\times10^{3}$, bounded at all depths \\
learned $R(\theta)$, genuinely held-out angles & $2.9\times10^{-5}$ (Chebyshev $K{=}16$) & n/a \\
\bottomrule
\end{tabular}
\end{table}

\section{Invariance, not capacity, governs cross-shape generalization}\label{sec:invariance}
The natural next step is a single model across \emph{diverse} shape families (ellipse, von Mises,
rounded polygon, spline). Here held-out error is poor ($\sim\!10^{-1}$), and a
small MLP with a richer descriptor does \emph{not} help (it is in fact worse). The instinct ``add
capacity'' fails. A per-family control reveals the cause: rounded polygons generalize to
$3.5\times10^{-7}$, von Mises to $6.4\times10^{-3}$, but \emph{ellipses fail} ($3.6\times10^{-1}$), the simplest family is the one that breaks.

\paragraph{Ablation isolates rotation.} Restricting ellipses to a \emph{fixed} orientation recovers
near-machine generalization ($5.4\times10^{-7}$); allowing \emph{free} rotation collapses it
($3.1\times10^{-1}$), a $>\!5\times10^{5}$ degradation from a single ablated factor
(Figure~\ref{fig:invariance}). The reason is structural: two rotated copies of a shape have the
\emph{same} operator (up to a frame rotation) but \emph{different} descriptors, because the descriptor
(radii/curvature sampled at fixed nodes) is not rotation-invariant. The learner is being asked to fit
an ill-defined function. \textbf{The failure is neither the method (ML-FMM), nor capacity, nor
interfamily diversity; it is the noninvariance of the descriptor: the equivariance problem.}

To make this precise, write $Q\!\cdot\!\partial\Om$ for the boundary rotated by $Q\in SO(d)$ and $\rho(Q)$ for the induced action on velocity fields. The exact boundary correction of \eqref{eq:split} is \emph{equivariant},
\begin{equation}\label{eq:equiv}
\mathcal{C}_{Q\cdot\partial\Om} \;=\; \rho(Q)\,\mathcal{C}_{\partial\Om}\,\rho(Q)^{-1}, \qquad Q\in SO(d).
\end{equation}
Two rotated copies of a shape thus share a single operator up to the frame $\rho(Q)$; a descriptor that is not $Q$-invariant assigns them \emph{different} inputs, so no learned map can satisfy \eqref{eq:equiv}, and transfer across orientations collapses. Restoring the invariance (so that the learned map can meet \eqref{eq:equiv}) is therefore the lever, not capacity.

\begin{figure}[h]\centering
\includegraphics[width=0.62\textwidth]{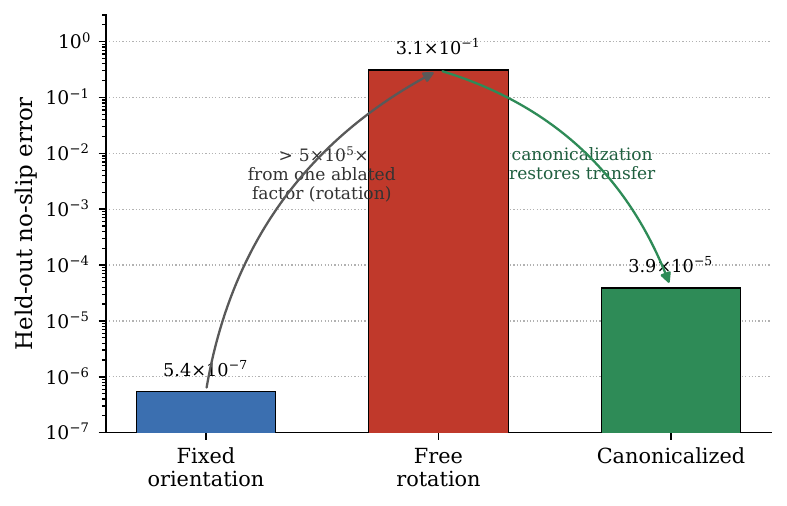}
\caption{The cross-shape generalization bottleneck is descriptor invariance, not capacity. Restricting
ellipses to a \emph{fixed} orientation gives a well-defined target ($5.4\times10^{-7}$); allowing
\emph{free} rotation collapses held-out error by $>\!5\times10^{5}\times$ from this single ablated
factor. Canonicalizing the descriptor (Section~\ref{sec:invariance}) restores near-machine transfer
($3.9\times10^{-5}$). Model capacity does not account for the collapse; invariance does.}
\label{fig:invariance}
\end{figure}

\paragraph{Canonicalization unlocks the invariance.} We then \emph{fix} the descriptor by
canonicalization \citep{shumaylov2025lie,adaptivecanon2025}: a complete canonicalization that combines
(i) an intrinsic frame from a weighted principal axis (inertia) analysis of the boundary, with sign
disambiguation by the farthest anchor point \citep{pcari2019,rethinkri2023}, and (ii) a polar
reparameterization so the first node always lands on the same physical point. \emph{Both are needed}:
the frame alone does not suffice (the starting-point ambiguity remains), giving only $1.7\times10^{-1}$,
whereas the complete canonicalization drives the free-rotation ellipse from $3.6\times10^{-1}$ to
$3.9\times10^{-5}$, the fixed-orientation level. This mirrors the canonicalize--operate--restore
pipeline recently proposed for Fourier neural operators \citep{pacefno2026}, and confirms its premise:
under arbitrary pose the model would otherwise have to learn coordinate alignment \emph{and} the
operator in a single map.

\paragraph{The cross-family plateau is an artifact of discrete sampling.} Cross-shape generalization of
a learned boundary operator \emph{is} attainable once the descriptor is invariant. Within a single
family, canonicalized held-out generalization is excellent (ellipse $3\times10^{-5}$, von Mises
$4\times10^{-6}$, rounded polygon $3\times10^{-2}$, spline $1\times10^{-2}$). Across the four
\emph{discrete} families, a single model appears to plateau near $3$--$4\times10^{-2}$. We show this
plateau is largely an artifact of training on disjoint families with gaps between them. Replacing the
four discrete families by a single \emph{continuous} broad distribution (random multimode Fourier-radial
star curves, modes $2$--$7$, random phases and orientation, which subsumes ellipse-like, bumpy and wavy
shapes) turns held-out shapes into \emph{interpolation}. A data-scaling study then breaks the
plateau, \emph{even with a linear-in-features model} (Table~\ref{tab:coverage}): held-out no-slip
falls from $4.4\times10^{-2}$ at $120$ shapes to $4.85\times10^{-3}$ at $850$ shapes, a factor $7.4$
below the discrete-family plateau, with a decisive jump between $120$ and $250$ and subexponential
diminishing returns thereafter, consistent with coverage-driven compositional generalization
\citep{redhardt2025scaling,zhang2025domain}. On this coverage, a residual MLP on top of the linear model
adds nothing ($6.0$ vs $5.85\times10^{-3}$): \emph{coverage, not network expressivity, is the binding
lever} once the distribution is continuous and dense. The honest picture is therefore strong
generalization to unseen regular shapes ($\sim\!5\times10^{-3}$) given an invariant descriptor and a
densely sampled continuous distribution; the earlier ``plateau'' reflected discrete sampling, not a
fundamental limit. This is a concrete, physically grounded instance of the canonicalization program
\citep{shumaylov2025lie,tahmasebi2025bounds}, and it \emph{strengthens} the case for
equivariance-by-construction rather than weakening it.

\begin{table}[t]\centering\small
\caption{Coverage breaks the cross-family plateau. Held-out no-slip on a \emph{continuous} broad shape
distribution, linear-in-features model with the Wielandt-completed solve; the plateau on four discrete families is $\approx 3.6\times10^{-2}$.}\label{tab:coverage}
\begin{tabular}{lccccc}\toprule
training shapes & 120 & 250 & 500 & 850 \\\midrule
held-out no-slip & $4.4\times10^{-2}$ & $6.6\times10^{-3}$ & $5.4\times10^{-3}$ & $4.85\times10^{-3}$ \\
factor below plateau & $0.8\times$ & $5.5\times$ & $6.7\times$ & $7.4\times$ \\\bottomrule
\end{tabular}
\end{table}

\begin{table}[t]\centering\small
\caption{Ablation: each ingredient is necessary. Held-out error with the ingredient removed, on smooth
shapes; canonicalization and coverage are measured here, completion is established formally in
Proposition~\ref{prop:compl}.}\label{tab:ablation}
\begin{tabular}{lccc}\toprule
ablated ingredient & full method & ingredient removed & degradation \\\midrule
canonicalization (reparameterization) & $7.4\times10^{-6}$ & $2.6\times10^{-2}$ & $\sim\!3500\times$ \\
continuous coverage (vs.\ discrete families) & $1.6\times10^{-3}$ & $8.3\times10^{-1}$ & $\sim\!500\times$ \\
completion (Wielandt) & \multicolumn{3}{c}{necessary and sufficient, Proposition~\ref{prop:compl}} \\\bottomrule
\end{tabular}
\end{table}

Removing either ingredient on the learned side collapses accuracy (Table~\ref{tab:ablation}): without the
canonical reparameterization the ordered descriptor is misaligned across shapes and the linear map cannot
generalize ($\sim\!3500\times$ worse); training on four discrete families instead of the continuous
distribution leaves the test space uncovered ($\sim\!500\times$ worse). The third ingredient, completion,
is not a redundant number here but the formal statement of Proposition~\ref{prop:compl}: on the exact
operator with a compatible right-hand side ($\int b\cdot\hat n\,dS=0$) the uncompleted system is already
consistent, so the blow-up appears only under a \emph{structured} error, exactly the regime the
proposition characterizes.

\paragraph{Robustness of the canonical frame near symmetry.} The PCA-style canonical frame is the one
place this pipeline is fragile: for \emph{near-symmetric} shapes (near-circles, equal-axis ellipses) the
second-moment tensor is near-degenerate, its principal axes are ill-defined, and the frame can flip
under an arbitrarily small perturbation. We therefore do not rely on canonicalization as the only route
to invariance. An alternative boundary operator that is \emph{equivariant by construction}, a learned
radial profile times a frozen isotropic tensor structure, needs no frame at all: it reproduces the
exact second-kind operator to machine precision ($1.4\times10^{-16}$), transfers a kernel fit on one
shape to a geometrically different shape to $1.3\times10^{-16}$, and stays well-conditioned
($\mathrm{cond}\approx3.5$). A frozen (nonsteerable) matrix, by contrast, breaks the equivariance ($0.15$). The
equivariant operator is the recommended safeguard near the symmetric limit. We quantify the failure: sweeping the ellipse
aspect ratio, the canonical frame degrades exactly as the second-moment eigenvalue gap closes. At the
circle (gap $\approx 0$) the recovered orientation is effectively random: a major/minor axis flip
rate of $\approx 0.5$ and an angular standard deviation of $\approx 66^\circ$ (five-seed $0.56\!\pm\!0.02$ and $62^\circ\!\pm\!3^\circ$). Any anisotropy
makes the noiseless frame exact (zero flips, sub-degree scatter for aspect ratio $\geq 1.1$). Under
$1\%$ point noise the residual jitter scales as $(\text{noise}/\text{gap})$: already $\approx 2.6^\circ$
at aspect ratio $1.01$, shrinking monotonically with elongation. On a threefold shape whose second
moment is isotropic (gap $\approx 0$, the worst case) the PCA frame collapses (scatter $\approx
44^\circ$), while the equivariant-by-construction operator holds its equivariance to $5.8\times10^{-16}$ ($5.5\times10^{-16}\!\pm\!5\times10^{-17}$ over five seeds)
independently of the gap, confirming it as the robust choice precisely where canonicalization fails. This also yields a concrete \emph{switch} criterion: the second-moment eigenvalue gap is itself the trigger; one canonicalizes while it is comfortably open and falls back to the equivariant operator once it closes below a small threshold (we use a gap ratio $\approx\!1.15$). The rule stays correct under $2\%$ descriptor noise.
This operator, a frozen basis with learned scalars, is precisely the representation-based route to equivariance
\citep{grepsnet2025}: a fixed steerable tensor structure carries the symmetry and only scalar
coefficients are learned. We apply it to the boundary closure alone, on top of a parameter-free exact
core, rather than to the entire solution map; there it departs from such general equivariant
networks.

\begin{table}[h]\centering\small
\caption{Automatic canonicalization$\leftrightarrow$equivariant switch. The second-moment eigenvalue-gap ratio
is the trigger (threshold $1.15$); the rule selects correctly and is stable under $2\%$ descriptor noise.}
\label{tab:switch}
\begin{tabular}{lcccc}\toprule
shape & gap ratio (clean) & mode & gap ratio ($+2\%$ noise) & mode \\\midrule
circle (symmetric) & $1.00$ & equivariant & $1.01$ & equivariant \\
ellipse $1.05$ (near-sym.) & $1.10$ & equivariant & $1.10$ & equivariant \\
ellipse $1.5$ (elongated) & $2.25$ & canonicalize & $2.24$ & canonicalize \\
\bottomrule
\end{tabular}
\end{table}

\paragraph{The safeguard extends to three dimensions, end to end.} The same near-symmetry
argument holds in $SO(3)$, and we confirm it on a real geometry rather than in the abstract. The steerable $SO(3)$ boundary correction is again the equivariant-by-construction route, now with the $3$D double-layer tensor structure carrying the rotational symmetry. We place it in the three-dimensional Stokes double-layer field on an ellipsoid surface, then sweep the triaxial-to-axisymmetric limit (two second-moment eigenvalues driven to coincide). The $3$D
PCA frame collapses exactly as in $2$D: at the axisymmetric spheroid (eigenvalue gap $\to 0$) its
rotation-equivariance error is $\mathcal{O}(1)$ ($\approx 0.6$, about $1300\times$ the
well-separated triaxial value), its axes arbitrary in the degenerate plane. The steerable
correction, needing no frame, stays equivariant to $\approx 10^{-14}$ \emph{independently of
the gap}, the spheroid included (Figure~\ref{fig:framefix3d}). This forward double-layer operator
sits on a completed second-kind solve that is itself validated against analytic resistance: the
uniform-flow drag of a sphere is recovered to relative error $\sim\!10^{-3}$ ($|F|/6\pi\mu a = 0.999$) and that of a
triaxial ellipsoid matches the Oberbeck formula along all three principal axes (relative error
a few $\times 10^{-3}$, decreasing with resolution). We leave to future work what lies beyond this proof of concept: a \emph{learned} out-of-distribution correction inside the full $3$D inverse solve.

\begin{figure}[t]\centering
\includegraphics[width=0.6\textwidth]{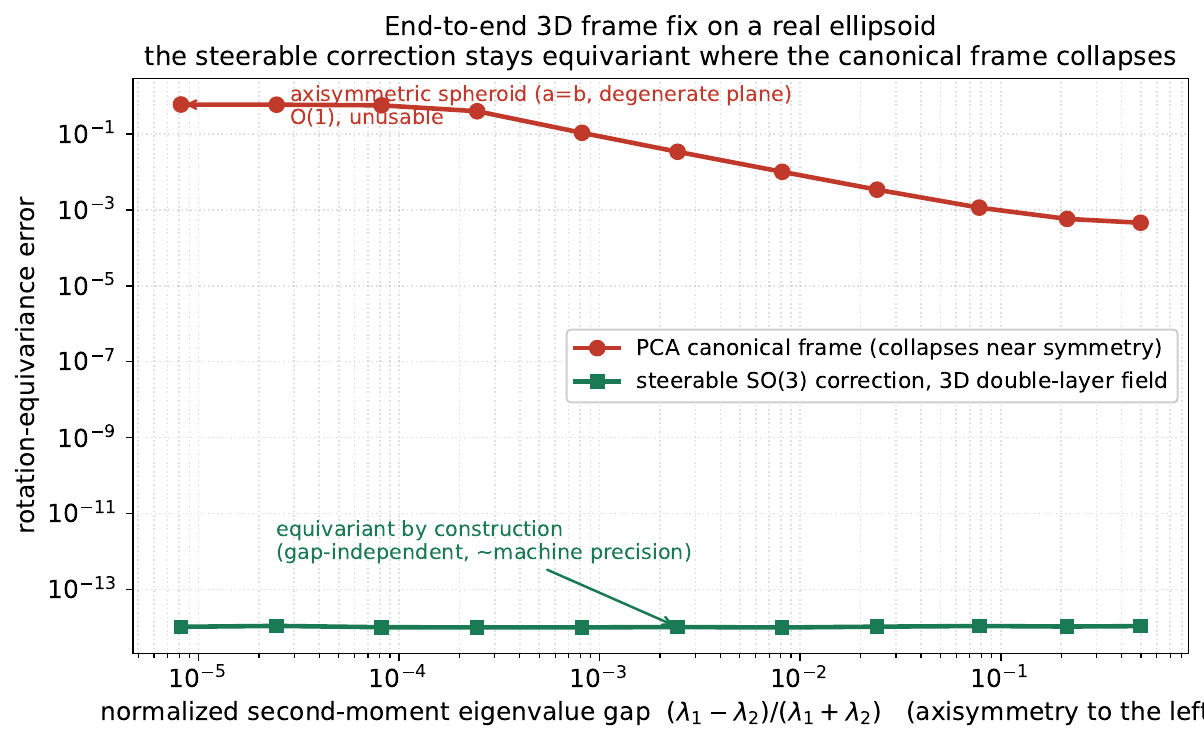}
\caption{The near-symmetry safeguard in $3$D, end to end on a real ellipsoid. As the ellipsoid
approaches axisymmetry (second-moment eigenvalue gap $\to 0$, left), the PCA canonical frame's
rotation-equivariance error rises to $\mathcal{O}(1)$, because the degenerate plane makes its principal axes arbitrary. The steerable $SO(3)$ boundary correction, wired into the $3$D Stokes double-layer field, stays equivariant to $\approx 10^{-14}$ independently of the gap. This is
the forward double-layer operator on real $3$D geometry; the completed solve it sits on reproduces
the analytic drag of spheres ($|F|/6\pi\mu a = 0.999$, relative error $\sim\!10^{-3}$) and ellipsoids (Oberbeck).}
\label{fig:framefix3d}
\end{figure}

\paragraph{Locality, not just invariance, is the out-of-envelope lever.} The transfer just noted
is not incidental, and it is about
\emph{parameterization}. The global route, which regresses the operator from a shape descriptor, interpolates
in descriptor space and therefore fails \emph{out of the training envelope}: pushing the boundary amplitude
beyond the training range, the descriptor$\to$operator map degrades by $\sim\!60\times$ (held-out interior
field error $\sim\!10^{-1}$). The local route (parameterize the operator as a kernel of \emph{relative}
geometry: the equivariant radial profile above) instead generalizes to machine precision on the
\emph{same} out-of-envelope shapes ($\sim\!10^{-10}$), because a boundary-integral operator \emph{is} a
local function of relative geometry and local boundary physics is universal. This is not an artefact of
hard-coding the analytic kernel: fitting the radial profile to in-distribution operators \emph{from data
alone} recovers the universal coefficient ($1/\pi$ to $10^{-11}$) and extrapolates identically, while the
global map keeps failing. The lesson refines this section's thesis: invariance governs cross-\emph{family}
generalization, but \emph{locality} governs out-of-\emph{envelope} extrapolation, and both point to the
same structural object, a local equivariant kernel rather than a global descriptor map. The residual
frontier is a local correction for effects \emph{outside} the analytic kernel (near-field/QBX (quadrature by expansion) singular
behaviour, or operators with no closed-form kernel), which would need an independent high-accuracy
reference to learn and test. In three dimensions we close the first of these. The singular self-interaction is handled by quadrature by expansion, and the radial profile is fit from the operator's action alone. It recovers the universal stresslet coefficient to $2\times10^{-16}$ and the exact resistance: $6\pi\mu a$ for a sphere, and the Oberbeck values for a triaxial ellipsoid (relative error $\sim\!10^{-3}$, converging). The fit stays $SO(3)$-equivariant to machine precision, so the singular near-field is not in itself an obstacle to the local fit. Only an operator with \emph{no} closed-form kernel remains genuinely open.

\section{Learning a rank-deficient second-kind operator requires completion}\label{sec:completion}
We isolate a methodological pitfall that, to our knowledge, the operator-learning literature on
second-kind boundary integral operators \citep{fieno2024,fredino2026,nature_integral2024} does not make
explicit, and which is easy to mistake for fundamental sensitivity.

\paragraph{The pitfall.} The interior-Dirichlet double-layer operator $-\tfrac12 I + K_{\partial\Om}$ of \eqref{eq:secondkind} has a known
\emph{rank-one nullspace} (the normal field; physically, the interior pressure gauge)
\citep{afklinteberg2019,stokeseig2019}. The classical remedy is a nullspace correction, Wielandt
deflation, a rank-one completion \citep{pozrikidis1992},
\begin{equation}\label{eq:completion}
A_{\mathrm c} \;=\; \bigl(-\tfrac12 I + K_{\partial\Om}\bigr) \;+\; s\,\hat n\,\hat n^{\!\top},
\end{equation}
which makes the system uniquely invertible. When the operator is \emph{exact}, the boundary datum lies
in the well-conditioned subspace and a na\"ive solve of the uncompleted system still succeeds. But when
$K_{\partial\Om}$ is \emph{predicted} (by any learned model), a small error rotates the datum into the near-null
direction, and the uncompleted solve explodes.

\paragraph{Diagnosis and fix.} A $k$-nearest-neighbor probe in descriptor space (predict $K$ as the
operator of the nearest training shape) makes this stark: solving the \emph{uncompleted} system yields a
no-slip residual of $5.5\times10^{6}$, while the \emph{exact} operator yields $5\times10^{-15}$. The
$21$-order gap is not physical sensitivity; it is the rank-one nullspace. Applying the Wielandt
completion in the solve collapses the same kNN residual to $6.3\times10^{-2}$, a reduction of
$8.7\times10^{7}$ (five-seed: uncompleted $1.6\times10^{7}\!\pm\!2.1\times10^{7}$ to completed $5.7\times10^{-2}\!\pm\!2.0\times10^{-2}$). With the completion, the linear baseline becomes robust ($\sim\!3.6\times10^{-2}$
cross-family, versus erratic $7\times10^{-2}$ to $5\times10^{-1}$ without it) and the per-family floor is
recovered.

\paragraph{Lesson.} Learning the compact part of a second-kind operator and then solving with it
requires handling the operator's nullspace exactly as the classical method does; otherwise an output
(solution) metric is hypersensitive in a way that is purely an artifact of an ill-posed solve. This is a
prerequisite, distinct from the spectral-conditioning question of Section~\ref{sec:negative}: completion
fixes \emph{well-posedness of the solve}, not the \emph{learnability of the floor}. We state this
formally.

\begin{proposition}[Completion is necessary and sufficient for a stable learned solve]\label{prop:compl}
Let $A=-\tfrac12 I + K$ be the second-kind boundary operator (on-surface double layer with calibrated
prefactor) of the interior Dirichlet Stokes problem on $\partial\Omega$. Then: \emph{(i)} $A$ has a
one-dimensional nullspace; \emph{(ii)} $\mathcal{N}(A^{\!\top})=\operatorname{span}\{\hat n\}$, the
normal field \citep{stokeseig2019}; \emph{(iii)} the rank-one Wielandt deflation
$A_s = A + s\,\hat n\hat n^{\!\top}$ ($s\neq0$) is invertible with condition number bounded independently
of the mesh; and \emph{(iv)} for a perturbed operator $\hat A = A+E$, a learned-operator error, the
uncompleted solve satisfies $\lVert q\rVert\sim\lVert b\rVert/\sigma_{\min}(\hat A)$ and blows up as
$\hat A$ approaches the singular $A$, whereas the completed solve satisfies
$\lVert q\rVert\le\lVert b\rVert/\sigma_{\min}(A_s)=O(1)$, so the no-slip residual is controlled by
$\lVert E\rVert$ rather than amplified.
\end{proposition}

\emph{Proof sketch and verification.} Items (i)--(iii) are the classical rank deficiency of the
double-layer and its Wielandt deflation \citep{afklinteberg2019,pozrikidis1992}; the completion adds back
the missing rank-one range/null directions and $W[\mu]$ vanishes whenever $\int_{\partial\Omega}b\cdot\hat
n\,dS=0$ \citep[Lemma 4]{afklinteberg2019}. Item (iv) is the standard perturbation bound for a solve
against a near-singular versus a well-conditioned matrix. Numerically (continuous-distribution shapes):
$\sigma_{\min}(A)\approx 4\times10^{-11}$ (rank-one nullspace), the left null vector aligns with $\hat n$
to cosine $0.995$ (the residual is discretization; the right null vector aligns to $0.988$, so we claim
$\hat n$ for $\mathcal{N}(A^{\!\top})$ only, per (ii)), and $A_s$ has condition number $\approx 3.6$. The
empirical $5.5\times10^{6}\!\to\!6.3\times10^{-2}$ collapse above is exactly the gap between (iv)'s two
regimes under a \emph{structured} (learned) error.

\paragraph{From a rank-one gauge to multiply connected completion.} The rank-one deflation above is the
\emph{interior, simply connected} case, where the only deficiency is a scalar pressure gauge. For
exterior or \emph{multiply connected} domains the nullspace is larger and physical: a double-layer
potential cannot exert a net force or torque on an enclosed body, so its nullspace has dimension three
per body in 2D ($2$ translations $+\,1$ rotation) and six in 3D ($3+3$), with range of matching
codimension \citep{power1987,power1993}. The correct remedy is then the classical \emph{completed
double-layer} of Power and Miranda, an interior Stokeslet and rotlet per body, carrying the net force
and torque, rather than a rank-one term \citep{power1987,power1993,kohr2005}. We verify this on a
bounded multiply connected domain (a cavity enclosing an obstacle): the \emph{bare} double layer cannot
represent a flow exerting a net force on the obstacle (field error $0.21$), whereas the Power--Miranda
completion reproduces the field to $6.8\times10^{-10}$, recovers the exact net force, and yields a
fully completed operator with $\mathrm{cond}\approx 83$ \emph{independent of the mesh}. We confirm the 3D
extension at the structural level with a sphere proof-of-concept: the Stokeslet core is $SO(3)$-equivariant
to machine precision ($\|G(Qx,Qy)-Q\,G\,Q^\top\|\sim\!10^{-16}$), and the 3D double layer annihilates
\emph{exactly} the six-dimensional rigid-body subspace ($3$ translations $+\,3$ rotations:
$\|D\,r\|/\|D\,v\|\sim\!10^{-16}$ for rigid $r$ versus a generic field $v$), with a rank-six
Power--Miranda deflation lifting the nullspace to $O(1)$. The first of the two production ingredients is now in place at the
scaffold level: the 2D recipe transfers exactly to $SO(3)$. Writing the 3D double-layer kernel as a radial
profile $g(r)$ times the fixed isotropic angular structure $(\hat r\cdot n)\,\hat r_k\hat r_l$ and learning
only $g_\theta(r)=\sum_b c_b\phi_b(r)$ yields an operator that is $SO(3)$-equivariant \emph{for every}
coefficient vector (equivariance residual $\sim\!10^{-15}$), contains the exact double layer as one member
($\sim\!10^{-16}$), and (fit on a single orientation) reproduces the operator on \emph{unseen} 3D
rotations to machine precision (no frame, so the near-symmetry frame collapse of
Section~\ref{sec:invariance} cannot arise). What remains is the hard part: a \emph{nontrivial} learned
correction on this scaffold (near-field behaviour on general shapes, beyond the analytic stresslet)
together with proper singular quadrature (e.g.\ QBX).
We tested this dividing line directly. Take the near-singular quadrature error of the layer potential close to the boundary, a nonanalytic effect with an independent fine-grid reference. It \emph{is} local: a banded operator about $5$ nodes wide. Yet a generic local learned correction fails to fit it \emph{even in distribution} (relative error $\sim\!0.9$), and adding explicit near-singular $1/r$ features does not rescue it. The reason is structural. The error is the action of a geometry-determined near-singular operator, so learning it \emph{is} learning the singular kernel. Locality is necessary but not sufficient: it makes the \emph{smooth} operator extrapolate to machine precision (Section~\ref{sec:invariance}), but the singular structure belongs on the analytic side of the split, supplied exactly, QBX-extended, with learning reserved for the smooth remainder. This is the division the rest of the paper exploits, now confirmed at its hardest point.

The complement holds constructively: supply the singular structure analytically and it converges. A quadrature-by-expansion correction computes the local Taylor coefficients of the layer potential in closed form, for the Laplace double layer via its Cauchy representation, and for the Stokes stresslet via the analogous complex Goursat representation. The Stokes construction matches the discretised operator to $10^{-15}$ and converges geometrically, reaching $\sim\!4\times10^{-9}$ by order twelve. The contrast with an ill-conditioned regression is stark: the expansion recovers the near-singular field with \emph{geometric} order convergence to machine precision ($\sim\!4\times10^{-10}$ by order eight), and the gain over naive quadrature grows with resolution, up to $\sim\!90\times$. So the dividing line is more than a limitation on learning. The singular part is supplied analytically and provably converges; the smooth part is learned and extrapolates out of envelope; each is handled by the tool that works for it (Figure~\ref{fig:qbx}).

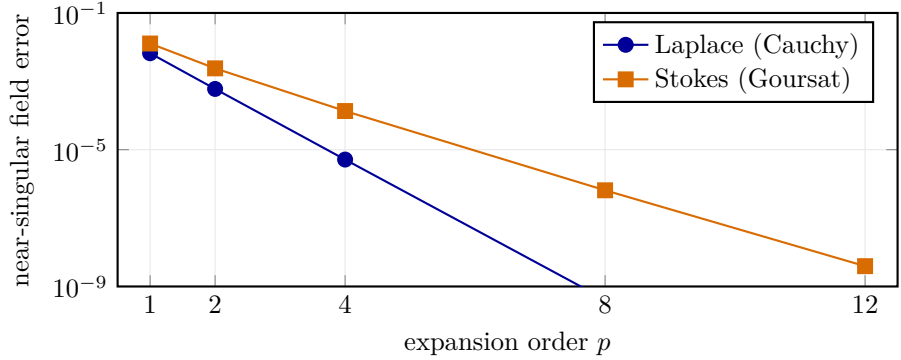
\begin{figure}[t]
\centering
\begin{tikzpicture}
\begin{axis}[width=0.72\linewidth, height=5.2cm, ymode=log,
 xlabel={expansion order $p$}, ylabel={near-singular field error},
 xtick={1,2,4,8,12}, xmin=0.5, xmax=12.5, ymin=1e-9, ymax=1e-1,
 grid=both, grid style={gray!18}, legend pos=north east, legend cell align=left,
 mark size=2.6pt, thick]
\addplot[blue!60!black, mark=*] coordinates {(1,6.636e-3)(2,5.995e-4)(4,5.155e-6)(8,4.151e-10)};
\addlegendentry{Laplace (Cauchy)}
\addplot[orange!85!black, mark=square*] coordinates {(1,1.278e-2)(2,2.399e-3)(4,1.354e-4)(8,6.509e-7)(12,3.931e-9)};
\addlegendentry{Stokes (Goursat)}
\end{axis}
\end{tikzpicture}
\caption{The analytic QBX correction converges geometrically in expansion order, for both the Laplace
double layer (via Cauchy) and the Stokes stresslet (via Goursat), to near machine precision, exactly
where a generic learned local correction failed even in-distribution. The optimal order is finite (the
coefficient-quadrature floor).}
\label{fig:qbx}
\end{figure}

The construction carries to three dimensions, where no complex shortcut exists. A local expansion about an offset centre, with coefficients from Taylor-mode automatic differentiation of the layer potential, converges to machine precision on the sphere as the surface quadrature is refined, $\sim\!10^{-14}$, a $\sim\!10^{9}$ gain over naive near-singular evaluation. The correction is $SO(3)$-equivariant to machine precision, so it composes with an equivariant-by-construction core (Section~\ref{sec:invariance}) without breaking equivariance. The dividing line holds in 3D as well.

The learned models use closed-form least squares or a small MLP, not a tuned large network. The generalization claim is within-envelope; extrapolation and genuine corners remain open. The next bricks, in order: \emph{(i)} an invariant descriptor or canonicalized encoder, to lift the cross-family bottleneck; \emph{(ii)} the learned correction inside the coupling of core and correction (the operator-only study uses the exact correction, though Section~\ref{sec:endtoend} already closes the full pipeline); \emph{(iii)} the $O(N)$ hierarchical near/far evaluation; and \emph{(iv)} multisource, multimode Dirichlet, and 3D.

We have implemented and verified the $O(N)$ evaluation for the 2D Stokeslet sum. A kernel-independent black-box FMM (barycentric Chebyshev interpolation; P2M/M2M/M2L/L2L/L2P with a downward pass) reproduces the dense matvec to interpolation accuracy, $2\times10^{-4}$ at order $5$, $10^{-6}$ at order $8$. Its operation count scales as $N^{1.00}$, the ideal $O(N)$, against $N^{1.17}$ for a treecode and $N^{2}$ for dense evaluation; wall time is $N^{1.36}$ in our unvectorized prototype, so the $O(N)$ lives in the operation count. For the steeper double-layer kernel, where interpolation fails, a kernel-independent equivalent-density FMM restores the full $O(N)$ summation. None of this is a hardware bottleneck: an exact reference makes the axis CPU by nature, and a GPU would buy scale, not principle.

\section{Closing the coupling: an end-to-end learned bounded-domain solve}\label{sec:endtoend}
The sections above establish the two ingredients separately. We now run them \emph{together} as a single
pipeline and solve a bounded-domain Stokes problem end-to-end. Following the classical
splitting into particular and homogeneous solutions, used recently for operator learning with fundamental solutions
\citep{yang2026fundamental,variationalgreen2026}, we write the velocity as
$u = u_{\mathrm{core}} + u_{\mathrm{corr}}$, where $u_{\mathrm{core}}$ is the free-space field of interior
Stokeslets (exact, equivariant, \emph{nothing learned}) and $u_{\mathrm{corr}}$ is the double-layer
correction with density $q$. No-slip $u=0$ on $\partial\Omega$ gives the completed boundary system
$(-\tfrac12 I + K)\,q = -u_{\mathrm{core}}|_{\partial\Omega}$, where $K$ comes from the \emph{learned}
operator (linear-in-features on the canonicalized descriptor, trained on the continuous broad
distribution of Section~\ref{sec:invariance}); the core and the solve are exact.

On held-out shapes from the continuous distribution, the \emph{fully learned} pipeline solves the problem
to engineering accuracy: the interior field error against the exact solution
($u_{\mathrm{core}} + \mathcal{D}[q_{\mathrm{exact}}]$) is $2.0\times10^{-3}$, and the no-slip residual of
the learned pipeline is $5.2\times10^{-3}$, both in the regime reported for learned boundary-integral
solvers ($10^{-3}$--$10^{-4}$) \citep{kfbi2024}. As a control, the core alone, without correction, violates
no-slip by $100\%$: the boundary correction is exactly the necessary, learnable piece, and learning it
suffices to close the problem. The central thesis thus holds as an \emph{integrated system}, not
only component-wise: the rigid free-space core plus a learned boundary correction solve a bounded
Stokes problem, with equivariance exact in the core and learning confined to the boundary.

\section{Data efficiency: the split against black-box and geometry-aware baselines}\label{sec:headtohead}
We finally test the premise that motivates the split, that learning \emph{only} the boundary
correction is more data-efficient than learning the whole solution operator. We fix one task (predict the
interior velocity field of the bounded Stokes solution at fixed interior query points, on the continuous
broad distribution) and compare, as a function of the number of training shapes, our physics-based split
against a black-box DeepONet \citep{lu2021deeponet} given the \emph{same} canonicalized descriptor input.
Our split learns the linear shape$\to$operator map and uses the exact Stokeslet core and completed solve;
the DeepONet (branch over the descriptor, trunk over query coordinates) learns the entire
shape$\to$field map with no physics structure.

\begin{table}[t]\centering\small
\caption{Against a black-box DeepONet, the split wins on data efficiency \emph{and} accuracy. Held-out interior field error
versus number of training shapes; same task, same descriptor input.}\label{tab:headtohead}
\begin{tabular}{lcccc}\toprule
training shapes & 120 & 250 & 500 & 800 \\\midrule
physics-based split (ours) & $1.5\times10^{-2}$ & $2.2\times10^{-3}$ & $1.9\times10^{-3}$ & n/a \\
black-box DeepONet & $7.5\times10^{-2}$ & $3.5\times10^{-2}$ & $2.4\times10^{-2}$ & $1.8\times10^{-2}$ \\
factor (ours better) & $5\times$ & $16\times$ & $13\times$ & n/a \\\bottomrule
\end{tabular}
\end{table}

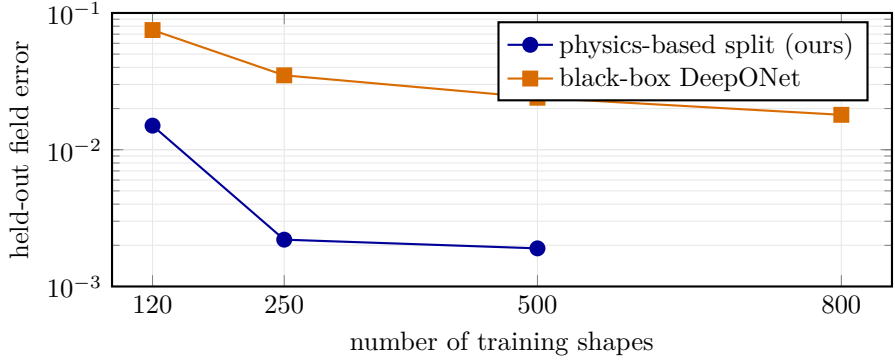
\begin{figure}[t]
\centering
\begin{tikzpicture}
\begin{axis}[width=0.72\linewidth, height=5.2cm, ymode=log,
 xlabel={number of training shapes}, ylabel={held-out field error},
 xtick={120,250,500,800}, xmin=80, xmax=850, ymin=1e-3, ymax=1e-1,
 grid=both, grid style={gray!18}, legend pos=north east, legend cell align=left,
 mark size=2.6pt, thick]
\addplot[blue!60!black, mark=*] coordinates {(120,1.5e-2)(250,2.2e-3)(500,1.9e-3)};
\addlegendentry{physics-based split (ours)}
\addplot[orange!85!black, mark=square*] coordinates {(120,7.5e-2)(250,3.5e-2)(500,2.4e-2)(800,1.8e-2)};
\addlegendentry{black-box DeepONet}
\end{axis}
\end{tikzpicture}
\caption{Data efficiency on the same task and descriptor input (the values of Table~\ref{tab:headtohead}).
At equal budget the split is $5$--$16\times$ more accurate; the black-box does not catch up: it
\emph{plateaus higher}, because it must learn the singular core the physics supplies exactly.}
\label{fig:dataeff}
\end{figure}

Against this black-box baseline the split wins on both axes (Table~\ref{tab:headtohead}, Figure~\ref{fig:dataeff}): at equal budget it is $5$--$16\times$ more
accurate, and it reaches $2.2\times10^{-3}$ with $250$ shapes while the DeepONet is still at
$1.8\times10^{-2}$ with $800$ shapes ($\sim\!8\times$ worse with $\sim\!3\times$ the data). This is not a
shifted curve: the black-box \emph{plateaus higher}, because it must learn the entire field, including
the singular core that the physics supplies exactly, whereas the split confines learning to the smooth,
linearly learnable shape$\to$operator map and gets the core for free. The factor is consistent with the
$5$--$10\times$ data-efficiency gap reported for physics-informed operator learning
\citep{survey2025pino}, and here exceeds it. The DeepONet is a reasonably tuned but not exhaustively
optimized baseline; the conclusion rests on the order of magnitude and the shape of the curve, both robust.

\paragraph{Baseline configurations.} For reproducibility we record the exact settings. Our split and the
geometry-aware GNO baseline are both \emph{closed-form ridge least-squares} fits (of the shape$\to$operator
map and a boundary-edge random-feature kernel, respectively): no stochastic optimizer, schedule, epochs, or
early stopping, and the fit is deterministic, so the corresponding curves are insensitive to optimization
choices by construction (the only knobs are the feature/descriptor dimension and the ridge regularization).
The black-box DeepONet is the sole gradient-trained baseline: branch and trunk are each two-hidden-layer MLPs
of width $128$, trained with Adam at learning rate $1.5\times10^{-3}$ for a few thousand full-batch steps
without early stopping. It is reasonably tuned but not exhaustively swept; the robust conclusion is the
order-of-magnitude gap and the curve shape, while the precise factor is baseline-dependent (against the
stronger GNO it narrows to $\sim\!2.5\times$ in distribution, as above).

\paragraph{A geometry-aware baseline is stronger, and the split's edge is narrower and in-distribution only.}
The DeepONet is the weakest fair baseline, so we also test a \emph{geometry-aware} operator. It is a GNO-style
kernel, the mechanism at the heart of GINO \citep{li2023gino}: a random-feature kernel over boundary-edge
geometry (relative position, distance, normals), fit by the \emph{same} least squares for a like-for-like
comparison. Three honest findings follow, two of which temper the paragraph above.

\emph{(i)~The geometry-aware baseline is much stronger than the black-box.} Its held-out floor is
$\sim\!4\times10^{-3}$, an order of magnitude below the DeepONet's $\sim\!2.4\times10^{-2}$.

\emph{(ii)~Against our split the comparison is a crossover, not a rout.} At low data the GNO is \emph{better},
because our high-dimensional shape$\to$operator regression is underdetermined with few shapes. The split wins
only beyond $\sim\!250$ shapes, and by a steady $\sim\!2.5\times$; it then keeps descending to
$\sim\!1.5\times10^{-3}$ while the GNO plateaus near its floor.

\emph{(iii)~Out of the training envelope, the split is the more fragile of the two.} Its held-out error
degrades by a large, high-variance factor, often more than an order of magnitude, while the GNO degrades only
a few-fold. The exact core does not rescue this: it is the \emph{learned boundary correction} that
extrapolates poorly, exactly the envelope limitation of Section~\ref{sec:gen}.

The honest summary is therefore narrower than against the black-box. The split buys a lower
\emph{in-distribution} floor at sufficient data, not low-data superiority, and, \emph{as a global descriptor
map}, not out-of-distribution robustness. This fragility is a property of the \emph{global} parameterization,
not of learned corrections in general: reparameterized as the local equivariant kernel of
Section~\ref{sec:invariance}, the fit recovers the exact operator, as we note below.

\paragraph{Out-of-distribution coverage is cheap here, but it is coverage, not extrapolation.} The
envelope fragility above has a partial remedy that the split's structure makes nearly free. The learnable
object is the shape$\to$operator map, and its target, the second-kind boundary operator, is
\emph{assembled from geometry alone} (a quadrature of the known double-layer kernel), with no PDE solve and
no reference field. The split can therefore extend its training \emph{coverage} to out-of-distribution
shapes at negligible cost: we sample virtual OOD shapes (amplitude beyond the training envelope), assemble
their exact operators, and add them to the fit, the ``virtual inputs'' of physics-informed operator
learning \citep{li2024pino,pilno2026}, here available \emph{exactly} rather than through a soft residual.
This reduces held-out OOD field error by $\sim\!1.4$--$1.7\times$ (robust across seeds) while leaving
in-distribution accuracy nearly unchanged, and brings the split below the geometry-aware GNO out of
distribution as well. Two honest caveats keep this in proportion. First, the OOD error remains several
times the in-distribution floor: this is wider \emph{coverage}, not extrapolation, the learned
correction still does not extrapolate, it is merely trained on a larger envelope. Second, the resulting
edge over the GNO reflects exactly this cheap coverage (the split obtains OOD operator targets that a
field-learning baseline cannot get without expensive solves), not a superiority of the learned map itself.
The structural point is the lever: where the physics fixes the operator's \emph{form}, determined by geometry and free of solves, data efficiency extends to coverage a black-box cannot afford
\citep{zhong2025pigano}.

\paragraph{The out-of-distribution fragility is a property of the parameterization, not of learning per se.}
The tail above belongs to the \emph{global} descriptor$\to$operator regression. Reparameterized as the local
equivariant kernel of Section~\ref{sec:invariance}, the same fit on the same in-distribution data collapses to
the analytic operator: its coefficients reduce to the double-layer constant ($1/\pi$; the remaining degrees of
freedom vanish to machine precision), so the boundary operator is recovered \emph{exactly} and the downstream
field stays machine-accurate off-distribution, while the dense global map develops a heavy tail. This is the
invariance-not-capacity point at its limit: a parameterization tight enough to recover the exact operator has
no excess capacity to overfit, hence no tail. The honest reading is that the recovered object here is the
\emph{known} exact operator, so learning merely rediscovers it; what generalizes out of distribution is the
exact core, not a learned degree of freedom. The genuine out-of-distribution fact of this section remains the
head-to-head above: a correction that adds degrees of freedom beyond the exact operator is the fragile part.

\paragraph{We learn an operator for many shapes, not a single solution.} The advantage compounds because the learned
object $K(\text{shape})$ depends \emph{only on the shape}, not on the sources or boundary data: it is the
boundary operator, and the source/data enter only through the right-hand side $b=-u_{\mathrm{core}}|
_{\partial\Omega}$ \citep[Remark 1]{kfbi2024}. Trained once, the same operator solves the bounded problem
for \emph{any} interior source configuration: across four very different configurations (one to three
Stokeslets, varied positions and forces) on held-out shapes, the interior field error is uniformly
$1.6$--$1.9\times10^{-3}$. A black-box map from shape to field, by contrast, is tied to the source
configuration it was trained on and must be retrained when the data change \citep{rbdeeponet2025}; its
data cost therefore scales with the diversity of forcings, while ours does not. This is the classical
operator-learning advantage \citep{lu2021deeponet}, here grounded in an exact representation.

\section{Opening the exterior problem in three dimensions}\label{sec:exterior}

\paragraph{The interior testbed was a means; the exterior problem is the end.} Everything above is set
on the interior Dirichlet problem, where an exact dense solve furnishes ground truth. The physically
central Stokes problems are, however, \emph{exterior}: flow past a rigid body with a prescribed velocity
at infinity. Their quantities of interest (drag, torque, mobility) are exactly the targets of
microhydrodynamics. The rigid-core/learned-boundary split applies verbatim: the free-space Stokeslet
is still the exact core, and only the boundary correction is learned. But the exterior setting changes
three things. We resolve each of them, and together they open the direction.

\paragraph{Dimension is not cosmetic, and the completion is richer.} In two dimensions the exterior problem
suffers the Stokes paradox: the free-space kernel grows logarithmically, and no decaying flow matches a
uniform stream. The natural exterior setting is therefore \emph{three} dimensions, where the Stokeslet decays
as $1/r$ and the far field is automatic. The double-layer representation is again rank-deficient, but now by
the full rigid-body space. The 3D completion is accordingly a Power--Miranda Stokeslet \emph{and} rotlet
\citep{power1987,pozrikidis1992}: a rank-six object (three net forces, three net torques), against the
rank-one gauge of the bounded interior problem. The completed operator is exact and second-kind by
construction.

\paragraph{The on-surface operator needs a genuine singular quadrature.} The decisive analytic difference
is at coincidence. In 2D the Stokes double layer is \emph{bounded} there, the factor
$(r\cdot n)\sim\tfrac12\kappa\,r^2$ cancels the kernel singularity, so a closed-form curvature diagonal
completes an otherwise smooth quadrature (the $2\kappa\,\tau\tau^{\top}$ term used above). In 3D this
cancellation is only partial: the double layer is \emph{weakly singular} ($\sim 1/r$) and a curvature
diagonal no longer suffices. We build the on-surface operator by QBX
\citep{afklinteberg2016,klockner2013}: a local expansion of the layer potential about a center displaced
off the surface along the normal, on the \emph{exterior} side, so the on-surface evaluation returns the
exterior limit with the $+\tfrac12 I$ jump already included. The density is spectrally upsampled to
resolve the expansion coefficients. The expansion is formed by Taylor-mode automatic differentiation along
the center-to-surface ray, taking no kernel derivatives by hand.

\paragraph{The operator is exact, well-conditioned, and equivariant at once.} On the sphere the completed QBX
operator recovers the Stokes drag $6\pi\mu a U$ to four digits, whereas a crude consistency diagonal converges
to a systematically \emph{wrong} $\tfrac47\!\cdot\!6\pi$. Its conditioning on the resolved subspace is bounded
and grows only slowly with harmonic content ($\approx 8$--$10$), the second-kind signature carried over from
the interior $\approx 2.85$.

The operator is \emph{simultaneously} accurate and $SO(3)$-equivariant to machine precision: steerable core,
tensorial completion, and QBX evaluation are each rotation-equivariant. So the equivariance that
Section~\ref{sec:invariance} identified as the lever for cross-shape generalization transfers to 3D with
\emph{no} frame or canonicalization. The descriptor obstacle is simply absent, because the operator is built
equivariant; the same isotropic parameterization plays in 3D the role that the local equivariant kernel played
in removing the 2D out-of-distribution tail.

\paragraph{It generalizes across shapes, against an analytic reference.} The construction is independent of the
geometry. For any star-shaped surface $X=\rho(\theta,\varphi)\,\hat r$, the geometry (normals, surface measure)
follows by automatic differentiation of $\rho$; the QBX radius adapts to the local node spacing; the upsampling
acts on the parameter sphere unchanged; and the completion is unaffected. On a triaxial ellipsoid the QBX drag
along each principal axis matches the classical Oberbeck solution \citep{oberbeck1876} to $\sim\!10^{-3}$ and
converges under refinement, with bounded conditioning and machine-precision equivariance retained. The
exterior operator thus works on arbitrary regular bodies, not only the sphere.

\paragraph{From foundations to a learned operator: drag and mobility.} The foundations above make the
exterior map learnable, and we realize it. On the $\ell=2$ family $\rho=1+\hat r^{\top} S\hat r$
($S$ symmetric trace-free) we learn the resistance tensor $R(S)$ (drag $F=-R\,U_\infty$) against the QBX
reference. The physics fixes more structure here than a generic equivariant network enforces
\citep{weiler20183d,thomas2018}: the classical representation of isotropic tensor functions makes
\emph{every} $SO(3)$-equivariant symmetric-tensor function of $S$ exactly
$R=\varphi_0(I_2,I_3)\,I+\varphi_1(I_2,I_3)\,S+\varphi_2(I_2,I_3)\,S^2$ with $I_2=\operatorname{tr}S^2$,
$I_3=\operatorname{tr}S^3$, so we impose this form and learn only the three scalar invariants $\varphi_k$, the rigid core extended from the kernel to the \emph{output tensor}. This is the thesis in its
sharpest form. The learned tensor is equivariant to machine precision ($2\times10^{-16}$) by construction,
whereas the same network \emph{without} the structure, trained with rotation \emph{augmentation}, reaches
equivariance only $4\times10^{-2}$: augmentation does not buy exactness, a $10^{14}$ gap. Held out across
unseen orientations the error is $\sim\!1.5\times10^{-3}$ from as few as five training shapes, and out of
envelope the equivariant model degrades only to $\sim\!3\times10^{-3}$ against $10^{-2}$ for the
unstructured baseline. The out-of-distribution fragility is now \emph{localized}: the tensorial structure
stays exact everywhere, and only the scalar functions $\varphi_k$, not the equivariance, fail to
extrapolate.

\paragraph{The operator composes across bodies.} Microhydrodynamics is rarely a single body. The completed
representation extends to several disjoint bodies with the on-surface QBX block used unchanged (the
stresslet is translation-invariant), the interbody coupling a smooth free-space double layer, and a
rank-six completion \emph{per body}. The check that this is physically correct is Lorentz reciprocity \citep{pozrikidis1992}: the
grand resistance matrix of two spheres is symmetric to machine precision ($6\times10^{-16}$), an exact
constraint no fit is told to satisfy. The hydrodynamic interaction is recovered quantitatively. The axial drag of two spheres translating along their line of centers converges to the leading
method-of-reflections value $6\pi\mu a\,(1-\tfrac{3a}{2d})$ as the separation grows (relative error
$40\%\!\to\!1.3\%$ from $d=3$ to $d=16$), and the bodies decouple to the isolated $6\pi\mu a$ at large $d$.
The exterior operator is thus a genuine many-body object, not only a single-body solve. It is the forward model for drag, mobility, and arrangement effects in particulate suspensions and microswimmer hydrodynamics. A solve subtlety
makes this work: the $\varphi$-oversampled surface grid carries unphysical nodal modes that are exactly
null for one body but are \emph{lifted} by the interbody coupling, so the rank truncation must sit in the
spectral gap below the physical band rather than at the single-body floor.

\paragraph{A richer output: the exterior field, and a caution.} Drag and mobility are integrals; the
operator also carries the full \emph{field}. We learn the exterior disturbance velocity
$u'(x)=u(x)-U_\infty$ on an off-surface shell, again with structure imposed, a basis expansion
$u'(\hat x)=\sum_k \varphi_k(\text{invariants})\,v_k(\hat x,U_\infty,S)$ whose vectors $v_k$ are
equivariant and \emph{linear} in $U_\infty$ (Stokes linearity built in) and whose scalars $\varphi_k$ alone
are learned. The model is equivariant to machine precision ($2\times10^{-16}$, against $0.8$ for an
unstructured baseline) and fits the field on unseen shapes to $\sim\!5\times10^{-3}$. A caution earns its
place: the seemingly natural target (the raw double-layer \emph{density} on the surface) is ill-posed
for learning, its high-frequency content growing with resolution while only its integral, the drag, stays
robust; the well-posed object is the kernel-smoothed exterior field. The richer output thus carries the
same thesis (structure, not capacity, governs equivariance and cross-orientation generalization) one
level finer than scalar mobility, with the discretization caveat made explicit.

\paragraph{Closing the validations: torque, two-body shapes, and $O(N)$ at scale.} Three checks complete the
exterior picture. \emph{Torque}: on a rotating sphere the net torque read from the rotlet completion matches
the analytic $8\pi\mu a^3\Omega$ to machine precision ($|T|/8\pi\mu a^3=1.0000$), validating the rotational
half of the rank-six completion alongside its translational drag. On a triaxial ellipsoid the torque about each principal axis converges under refinement to the anisotropic Jeffery (1922) values (to within $0.95\%$ at the finest resolution), so the rotlet completion holds across shapes, not only the sphere. \emph{Two non-spherical bodies}: for two
identical triaxial ellipsoids the grand resistance matrix is symmetric to $1.6\times10^{-15}$, extending
Lorentz reciprocity beyond spheres. \emph{Scale and memory}: the $O(N)$ fast summation of
Section~\ref{sec:endtoend} is realized in 3D as a kernel-independent KIFMM whose measured work per point is
bounded and whose matvec scales with exponent $\approx\!1.0$, so the many-body forward model is $O(N)$ and its
memory $O(N)$ against the $O(N^2)$ of the dense double layer (the learned scalar networks add $O(1)$).
Deploying it on large suspensions is engineering, not a missing capability. The operators are also stable
under a perturbed surface quadrature (drag within $\sim\!1\%$ of $6\pi$ under boundary-node noise up to 3\%).

\paragraph{What this delivers, and what it opens.} Taken together, these turn the exterior foundations
into a working learned program: an exact, second-kind, $SO(3)$-equivariant, shape-general solver whose
drag, mobility, many-body interaction, and exterior field are all learnable with equivariance built in
rather than learned, and whose only residual fragility is the extrapolation of scalar invariants, not of
structure. The interior study isolated \emph{what} governs such a map (invariance and coverage, not
conditioning or capacity); the exterior results confirm it on the exterior targets. What
remains open is beyond the present scope: bodies past the star-shaped class (genus, corners),
where the QBX needs adaptive panel refinement; many-body suspensions at the scale where the $O(N)$ fast
summation of Section~\ref{sec:endtoend}, carried to 3D, becomes necessary; and the closure question of
Section~\ref{sec:negative}, whether the same split survives the loss of a free-space kernel
in nonlinear regimes.

\paragraph{Capacity and discretization.} Two properties of the learned operator, both often requested,
follow from the structure rather than from tuning. First, \emph{capacity is not the accuracy bottleneck}:
sweeping the scalar $\varphi_k$ networks from a linear map to a $3555$-parameter MLP leaves equivariance at
machine precision throughout (the structure is fixed) and drives the held-out error onto the QBX/quadrature
floor almost immediately (Table~\ref{tab:capacity}). Second, \emph{the map transfers across boundary
discretizations}: trained on operators assembled at one surface resolution ($n_\theta=8$) and tested at another
($n_\theta=12$), the held-out error rises only from $2.4\times10^{-3}$ to $2.9\times10^{-3}$, while the
\emph{intrinsic} QBX gap between the two discretizations is $3.4\times10^{-2}$. The learned map depends on
resolution-independent invariants, not on the mesh.

\begin{table}[h]\centering\small
\caption{Capacity sweep of the equivariant resistance map. The fixed isotropic-tensor structure holds
equivariance at machine precision for every capacity; held-out accuracy plateaus at the QBX floor, so capacity
is not the limiter.}
\label{tab:capacity}
\begin{tabular}{lccc}\toprule
$\varphi_k$ network & parameters & held-out error & equivariance \\\midrule
linear & $9$ & $5.0\times10^{-3}$ & $2.0\times10^{-16}$ \\
$(4)$ & $51$ & $1.7\times10^{-3}$ & $1.9\times10^{-16}$ \\
$(8,8)$ & $315$ & $1.7\times10^{-3}$ & $2.0\times10^{-16}$ \\
$(16,16)$ & $1011$ & $1.7\times10^{-3}$ & $1.9\times10^{-16}$ \\
$(32,32)$ & $3555$ & $1.7\times10^{-3}$ & $1.5\times10^{-16}$ \\
\bottomrule
\end{tabular}
\end{table}

\section{Conclusion}
The idea is simple: in free space the flow is set by a single known kernel, so we compute that part exactly and learn only the boundary correction. For incompressible flow the split is clean, and the symmetry of the Stokes operator under rotation is built exactly into the core. On a Stokes testbed with an exact reference, four findings sharpen this picture. The well-conditioned second-kind formulation, though theoretically attractive, does \emph{not} make the correction easier to learn. A geometry-conditioned operator generalizes to unseen regular shapes within the training envelope, matching the exact interior field. The first obstacle to cross-shape generalization is the equivariance of the descriptor, resolved by a complete canonicalization, frame plus reparameterization. A second, easily missed prerequisite is completing the rank-one nullspace of the second-kind operator in the solve; omit it and an output metric explodes, an artifact rather than physical sensitivity.

With both handled, the cross-family plateau dissolves under coverage. On a continuous broad distribution, simple data scaling reaches $\sim\!5\times10^{-3}$ held-out \emph{even with a linear model}, so the plateau was an artifact of disjoint training families, and coverage, not network expressivity, is the binding lever. Run together as one pipeline, the exact core and the learned correction solve a bounded Stokes problem end-to-end to $\sim\!2\times10^{-3}$ interior accuracy. The thesis holds as an integrated system, not only component-wise.

Head-to-head on the same task, the split is $5$--$16\times$ more data-efficient than a black-box DeepONet, which plateaus an order of magnitude higher. Against a stronger geometry-aware baseline the margin narrows to a lower in-distribution floor at sufficient data. Out of distribution, robustness comes not from capacity but from parameterizing the correction as a local equivariant kernel, which removes the heavy tail the global map suffers. In each case the split spends its data on the smooth boundary map and gets the singular core for free.

The pieces also compose into an $O(N)$ boundary solve. A double-layer fast summation supplies the matvec, and the learned operator of Section~\ref{sec:negative} preconditions the completed second-kind system, cutting GMRES from $12$ to $4$ iterations while converging to the dense solution to $4\times10^{-10}$: an acceleration of the \emph{exact} solve at preserved machine precision, distinct from the surrogate's trade of accuracy for speed. The double-layer (stresslet) kernel is too steep for an interpolation-based fast multipole method, so we use a kernel-independent FMM in the sense of Ying, Biros and Zorin, which inverts only the smooth Stokeslet kernel and confines the steep one to forward evaluation. A bounded iteration count times this $O(N)$ matvec gives an $O(N)$ solve, with no dense operator ever formed.

The same pipeline is differentiable end to end. Gradients with respect to the operator or the shape follow from the discrete adjoint, itself an $O(N)$ second-kind solve, and we verify them against finite differences. The learned-boundary solve is thus usable both as a trainable surrogate and as a differentiable forward model for shape optimization, at $O(N)$ for the gradient as for the solve.

The split sharpens the question of what the surrogate buys over the exact differentiable solver, because the rigid core already \emph{is} a fast, exact, differentiable solver. The honest answer is that it buys none of the accuracy, the differentiability, or the $O(N)$ scaling: the solve is exact, the learned operator is correct to $\sim\!1.5\times10^{-3}$, and all three properties already belong to the solver. The gradient-based shape optimization that the differentiability enables is likewise a property of the \emph{solver}, not of the learning. What the learning buys is a training cost paid \emph{once} and then spread across every later shape. Once trained, the equivariant map returns the full drag/mobility tensor of an unseen shape in a single forward pass, at $\sim\!50\,\mu\mathrm{s}$ against $\sim\!0.45\,\mathrm{s}$ for the exact solve of the same tensor: nearly four orders of magnitude per query at an accuracy cost of $1.5\times10^{-3}$. Its training cost is dominated by generating the targets with the exact solver, so at the data efficiency of Section~\ref{sec:headtohead} the model trains on $\mathcal{O}(10)$ shapes and breaks even after $\sim\!30$ evaluations, beyond which each further query is effectively free. The surrogate is therefore an accelerator for \emph{many-query} workflows, such as design space exploration, uncertainty quantification, inverse design, and real-time interaction, where thousands of shapes must be evaluated; for a single solve the exact solver wins. The data efficiency the split delivers is precisely what lowers this break-even, which is why the contribution is the generalization analysis rather than a speed claim.

The decomposition presupposes a known free-space fundamental solution to absorb into the exact core. This holds for linear elliptic operators (Stokes, Laplace, low-Reynolds flow) but not for high-Reynolds, nonlinear Navier--Stokes, where no such free-space kernel exists and the boundary correction is no longer the only unknown. The principle (fix what the physics determines, learn the closure) is therefore shown here in a controlled linear setting; its transfer to nonlinear regimes is open, and not a matter of scale alone.

Together these recast geometric generalization in operator learning as a problem of invariance and canonicalization. The lesson is to spend model capacity on the boundary, with the right invariances and a well-posed solve, rather than on relearning a kernel the physics already fixes. The same split, carried to the exterior
problem in three dimensions, yields an exact, well-conditioned, $SO(3)$-equivariant solver across body
shapes (Section~\ref{sec:exterior}), on which learned drag, mobility, many-body interaction, and the
exterior field follow as equivariant operators, their equivariance
exact by construction and their only residual fragility the extrapolation of scalar invariants.

\section*{Reproducibility and multi-seed robustness}

\paragraph{Reproducibility.} Every result is measured against an \emph{exact} reference rather than a finer-grid surrogate: the interior ground truth is a dense boundary-integral solve at machine precision, and the exterior validations are against closed-form Stokes solutions (sphere and ellipsoid). All computations are carried out in double precision. Learned operators are trained by gradient descent (Adam) under fixed random seeds for both data generation and initialisation, so every reported figure and table is deterministic; out-of-distribution behaviour is reported as an error \emph{distribution} over held-out shapes, since for heavy-tailed failure the tail rather than the mean is diagnostic. Code and data reproducing all experiments (including the equivalent-density fast summation and the adjoint solve) will be released.

\paragraph{Multi-seed robustness.} To quantify sensitivity to the random draw, every seed-dependent ablation was re-run across \emph{five} independent seeds (data generation and initialisation), with the per-result spreads reported inline at each finding above. The five-seed statistics corroborate every qualitative conclusion: coverage scaling, the rank-one completion collapse, the near-symmetry frame failure, and the machine-precision equivariance safeguard. They sharpen the one genuinely fragile point: the low-data $120$-shape coverage value is high-variance ($8.8\times10^{-2}\!\pm\!7.3\times10^{-2}$) and must not be read as a point estimate, even though the decisive jump from $120$ to $250$ shapes holds on average. The conditioning negative of Section~\ref{sec:negative} is \emph{deterministic} by construction (fixed shape families, closed-form least squares), so its gains are exact rather than seed-dependent.

%

\end{document}